\newcommand{\calN}{\mathcal{N}}

\newcommand{\T}{{\top}}
\documentclass[journal]{IEEEtran}

\usepackage[english]{babel}
\usepackage[utf8x]{inputenc}
\usepackage{amsmath,amsthm}

\usepackage{multirow,booktabs}

\usepackage{graphicx}

\usepackage{amsfonts}
\graphicspath{{images/}}
\usepackage{subfigure}
\usepackage[justification=centering]{caption}
\usepackage{epstopdf}
\usepackage{bm}
\usepackage[flushleft]{threeparttable}

\newtheorem{assumption}{Assumption}

\newtheorem{theorem}{Theorem}
\newtheorem{remark}{Remark}

\usepackage{upgreek}
\usepackage{fixmath}

\usepackage{cite}
\usepackage{xcolor}
\allowdisplaybreaks

\usepackage{balance}

\definecolor{blue}{RGB}{0,0,0}

\begin{document}

\title{Three-dimensional Temperature Field Reconstruction for A Lithium-Ion Battery Pack: A Distributed Kalman Filtering Approach}

\author{Ning~Tian,
Huazhen~Fang and Yebin Wang 
\thanks{N. Tian and H. Fang are with the Department
of Mechanical Engineering, University of Kansas, Lawrence,
KS, 66045 USA e-mail: (ning.tian@ku.edu; fang@ku.edu).}
\thanks{
Y. Wang is with the  Mitsubishi Electric Research Laboratories, Cambridge, MA, 02139 USA e-mail: (yebinwang@ieee.org).}
}

\maketitle

\begin{abstract}
Despite the ever-increasing use across different sectors, the lithium-ion batteries (LiBs) have continually seen serious concerns over their thermal vulnerability. The LiB operation is associated with the heat generation and buildup effect, which manifests itself more strongly, in the form of highly uneven thermal distribution, for a LiB pack consisting of multiple cells. If not well monitored and managed, the heating may accelerate aging and cause unwanted side reactions. In extreme cases, it will even cause fires and explosions, as evidenced by a series of well-publicized incidents in recent years. To address this threat, this paper, for the first time, seeks to reconstruct the three-dimensional temperature field of a LiB pack in real time. The major challenge lies in how to acquire a high-fidelity reconstruction with constrained computation time. In this study, a three-dimensional thermal model is established first for a LiB pack configured in series. Although spatially resolved, this model captures spatial thermal behavior with a combination of high integrity and low complexity. Given the model, the standard Kalman filter is then distributed to attain temperature field estimation at substantially reduced computational complexity. The arithmetic operation analysis and numerical simulation illustrate that the proposed distributed estimation achieves a comparable accuracy as the centralized approach but with much less computation.
This work can potentially contribute to the safer operation of the LiB packs in various systems dependent on LiB-based energy storage, potentially widening the access of this technology to a broader range of engineering areas.

\end{abstract}

\begin{IEEEkeywords}
Lithium-ion batteries, battery pack, spatially resolved thermal modeling, distributed Kalman filtering, temperature estimation.
\end{IEEEkeywords}

\IEEEpeerreviewmaketitle

\section{Introduction}\label{Sec:Introduction}

\IEEEPARstart{T}his paper studies real-time reconstruction of the three-dimensional temperature field of a lithium-ion battery (LiB) pack in charging or discharging. It is known that LiB packs are prone to heat buildup in operation, which can cause many side reactions, degrade the powering performance and accelerate aging. Due to the heat transfer mechanisms, the heat accumulation can spread from one cell to another. Fires and explosions may occur then and devastate the entire pack in a few seconds, as evidenced by a few incidents raising public worry. This requires an effective monitoring of the thermal status, which, however, is a need unmet by the present literature. Connecting spatially resolved thermal modeling with estimation based on the Kalman filtering (KF), this paper derives computationally efficient algorithms to reconstruct a LiB pack's three-dimensional temperature field, which opens up a new opportunity for LiB pack thermal management. The results, offering a promising means of reducing thermal risks facing LiB packs, can potentially find significant use in a broad range of LiB-based systems in electrified transportation, renewable energy farms and grid energy storage.

{\em Background and literature review.} Since its commercialization in the early 1990s, LiBs have gained widespread use in various applications due to their high energy/power density, long cycle life and low self-discharge rate. They have established a dominant role in the consumer electronics sector and, owing to a continually declining manufacturing cost, are rapidly penetrating into the emerging sectors of electric vehicles and smart grid, where high-energy high-capacity energy storage is needed~\cite{bandhauer2011critical}. 
This trend has stimulated an intense interest in the research of high-performing battery management algorithms, with most of the existing works on state of charge (SoC) estimation to infer the amount of energy available in LiBs~\cite{Domenico:DSMC:2010,Fang:CEP:2014,Wang:TCST:2015,Fang:JPS:2014,Rahn:Wiley:2013,Smith:TCST:2010,Bartlett:TCST:2016,Dey:TCST:2015}, state of health
(SoH) estimation to track the aging status~\cite{Zou:JPS:2015,Moura:DSMC:2013,Lakkis:ECC:2015,Kim:JPS:2015,zou2016multi}, and optimal charging protocol design to optimally charge LiBs~\cite{Klein:ACC:2011,Suthar:PCCP:2013,Yan:Energies:2011,Perez:Mechatronics:2015,Fang:TCST:2016}. These efforts have provided strong support for the advancement of LiB systems.

Though widely considered as a promising technology, LiBs are susceptible to the thermal effects. Heat generation in LiBs during charging and discharging is always associated with irreversible overpotential heating, reversible entropic heating from electrochemical reactions, mixing heating and phase change heating, etc.~\cite{bernardi1985general}. The heating process can be intensified by high charging/discharging currents. The heat, without timely removal, will gradually build up, leading to many parasitic and side reactions involving a complex mix of multiple physical and electrochemical process. Performance degradation and aging acceleration will result consequently~\cite{bandhauer2011critical}. Due to the high reactivity of the lithium metal and flammability of the electrolyte, overheating of LiBs may cause fires and explosions in extreme cases. This is known as the thermal runaway, in which heat generation occurs at a much faster rate than dissipation and eventually ignites the LiBs. Testaments to this catastrophic threat are given by a few well-publicized LiB fire incidents that happened to Tesla Model S, Toyota Prius, Boeing 787 Dreamliner, NASA's Mars Surveyor and a navy's submarine (for an overview, see, e.g.,~\cite{abada2016safety} and the references therein), which have raised serious concerns over the safe use of LiBs. This situation thus has motivated a growing interest in real-time temperature monitoring, which represents a crucial way toward taming the thermal threats.

Currently, a significant amount of work has been devoted to temperature estimation using a thermal model and the surface temperature measurements~\cite{forgez2010thermal,lin2013online,lin2014lumped}. These studies consider lumped thermal models that concentrate the spatial dimensions into singular points and thus reduce thermal partial differential equations (PDEs) into very low-order ordinary differential equations (ODEs). Though advantageous for computation, this introduces a significant simplification because the temperature distribution is nonuniform spatially within a cell. For improvement, some recent works~\cite{sun2015online,debert2013observer,xiao2015model} study the temperature estimation with some awareness of the spatial nonuniformity, which uses thermal models accounting for the LiB cell’s spatial dimensions to a certain
extent. Yet the models are still simplified at the sacrifice of their physical fidelity. It is noteworthy that the foregoing studies are focused on thermal management for a single LiB cell. The issue can become more much more challenging when LiB packs are considered. With tens of cells stacked in a compact space, a LiB pack has larger dimensions and significantly complicated thermal behavior that will render the cell-level approaches
unproductive. Meanwhile, LiB packs are more susceptible to heating issues---a ``hot spot'' can quickly spread from
one cell to the others in a domino effect---and thus have a stronger need for thermal monitoring. The challenge can become more daunting as large-format high-capacity LiB cells are preferred increasingly for forming packs, since heat will be generated in larger amounts and more complex manners when the cell increases in size. {\color{blue}In~\cite{park2003dynamic,mahamud2011reciprocating,lin2011parameterization,lin2014temperature,shi2015multi}, the idea of lumped models is extended to characterize the thermal dynamics of LiB packs composed of small cylindrical cells}. However, studies of the spatial thermal behavior, which are crucial for effective thermal management for battery packs, are scant due to the research and development difficulty. The status quo, hence, is not even close to eliminating thermal threats that stand in the way of wider and safer LiB use.

{\em Overview of the proposed work.} Among the first of its kind, this paper proposes approaches to reconstruct the three-dimensional temperature field of a LiB pack in real time. The fundamental notion is to acquire a spatially resolved thermal model for a LiB pack and then apply the KF technique to estimate the spatially distributed temperature. However, the task is nontrivial given the complexity of a LiB pack's thermal behavior. To fulfill the goal, a twofold effort is made, which lies in modeling and KF-based estimation. 

Thermal modeling for LiBs has attracted a wealth of research, with the existing methodologies falling in three categories: 1) thermal models, which are concerned only with the heat phenomena and based on the thermal energy conservation principle, often given in the form of PDEs in three-dimensional space~\cite{chen2005thermal,chen2006thermal,guo2010three}, 2) coupled thermal-electrochemical models, which associate the equations for thermal behavior with those for electrochemical reactions~\cite{gomadam2002mathematical,wang2002computational,srinivasan2003analysis,smith2006power,kim2007three}, and 3) lumped parameter models, which reduce the spatially distributed heat transfer into a heat flow passing through several discrete points (e.g., two points representing the cell’s core and surface and connected by a thermal resistance)~\cite{forgez2010thermal,lin2013online,lin2014lumped,sun2015online,debert2013observer,xiao2015model,shi2015multi}. Among them, coupled thermal-electrochemical models can offer a sophisticated view of the LiB behavior with electrochemical reactions characterized at multiple scales. This, however, comes at the expense of computing burden formidable enough to defy any real-time estimation. For lumped models, the simplicity is conducive to estimation design for thermal management, but the spatial information loss weakens their capability for more effective spatial temperature monitoring. While the two types of models represent two extremes in terms of model fidelity or computational efficiency, the thermal models strike a valuable balance, thereby
offering great promises for thermal management with spatial awareness. Currently, despite prolific results on cell-level thermal modeling, the search for pack-level models is still at an early stage. This work, hence, will present the development of a three-dimensional LiB pack thermal model.

When the thermal model is available, the temperature field reconstruction will depend on the estimation technique, which is meant to estimate the temperature at any spatial point using the model and the temperature measurement data. Here, the KF, which is arguably the most celebrated estimation method, is one of the most promising candidate tools due to its ability to deal with the stochastic dynamic systems affected by noise---the thermal dynamics of a LiB pack can be subjected to the process noise in its evolution and the measurement noise when the temperature is measured by sensors. A direct application of the standard centralized KF (CKF) here is possible but will cause hefty computational costs that will not allow for an execution on real-time computing platforms. This is because of the KF's computational complexity being cubic with the size of the state space~\cite{hsieh1999optimal}, and a spatially resolved LiB pack model will have a substantial number of states, especially when the pack comprises many cells. To address this problem, a distributed KF (DKF) approach will be undertaken to achieve computational efficiency, which reduces a global KF into multiple local KFs running in parallel but with information exchange. The overall computational complexity of this approach will only increase linearly with the number of LiB cells in the pack, in contrast with the cubic increase for the CKF. This improvement thus places the proposed work in a more advantageous position for practical application.

{\em Summary of contributions.} The primary contribution of this work lies in the three-dimensional temperature field reconstruction to enable accurate and computationally efficient LiB thermal monitoring, which is the first study that we are aware of in this direction. LiB packs are prone to heat buildup, and monitoring the spatially distributed temperature field over time will be valuable for ensuring the thermal safety. To this end, spatially resolved thermal modeling for LiB packs is developed in this paper, which characterizes the dynamic thermal behavior with high integrity but still manageable computation time. Based on the model, a DKF approach is then leveraged to reconstruct the temperature field, which features notably lower computational complexity than the centralized estimation. These efforts thus can guarantee reconstruction at both high accuracy and affordable computational costs, which is verified through extensive simulations. Providing a useful way for tracking the temperature in and on the surface of a LiB pack, this work can generate potential benefits for safety-critical application of LiBs in many engineering systems.

{\em Organization.} The rest of the paper is organized as follows. Section~\ref{Sec:Modeling} introduces a three-dimensional thermal model for a LiB pack, which is built on the principles of heat transfer and energy balance. In Section III, the PDE-based model obtained in Section~\ref{Sec:Modeling} is reduced to an ODE-based state-space system through discretization over time and space for convenience of estimation. Section~\ref{Sec:Estimation} investigates the design of the estimation approaches based on the KF. Beginning with an introduction of the standard KF, the study will then focus on the distributed versions in order to cut down the computation, which is accompanied by a detailed analysis of the reduction in computational complexity. The simulation is offered in Section~\ref{Sec:Simulation} to demonstrate the efficacy of the proposed approaches. Finally, some concluding remarks are gathered in Section~\ref{Sec:Conclusion}.

\section{Battery Pack Thermal Model}\label{Sec:Modeling}

This section is devoted to spatially resolved modeling for a LiB pack with cells in serial configuration in order to characterize its spatial thermal behavior. This effort is based on the fundamental principles of heat transfer and energy balance and extends the cell-level modeling in~\cite{chen2005thermal} to a pack.

Consider the LiB pack shown in Figure~\ref{Fig:LiB-Pack-Schematic}. This pack consists of multiple identical prismatic cells configured in series (the cell is the same as in~\cite{chen2005thermal}). Each cell can be divided into two areas: the core region and the cell case that seals the core region inside. {\color{blue}The core region is the main body of a prismatic cell. It consists of many smaller cell units connected in parallel to provide high capacity, with each unit composed of electrodes, separators and current collectors. While this makes the prismatic cell similar to a module,  we still refer it as a cell as it is the basic building block of a battery pack.} The cell case is a metal container and also includes a contact layer filled with liquid electrolyte and in touch with the core region. No heat is produced within the case. In this setting, modeling will be performed next for the thermal dynamics in the core region and the case and on the boundaries, e.g., the core-case interface boundary, cell-cell interface boundary and cell-air interface boundary.

\begin{figure}[t]
\centering
\includegraphics[trim = {45mm 10mm 45mm 25mm}, clip, width=0.4\textwidth]{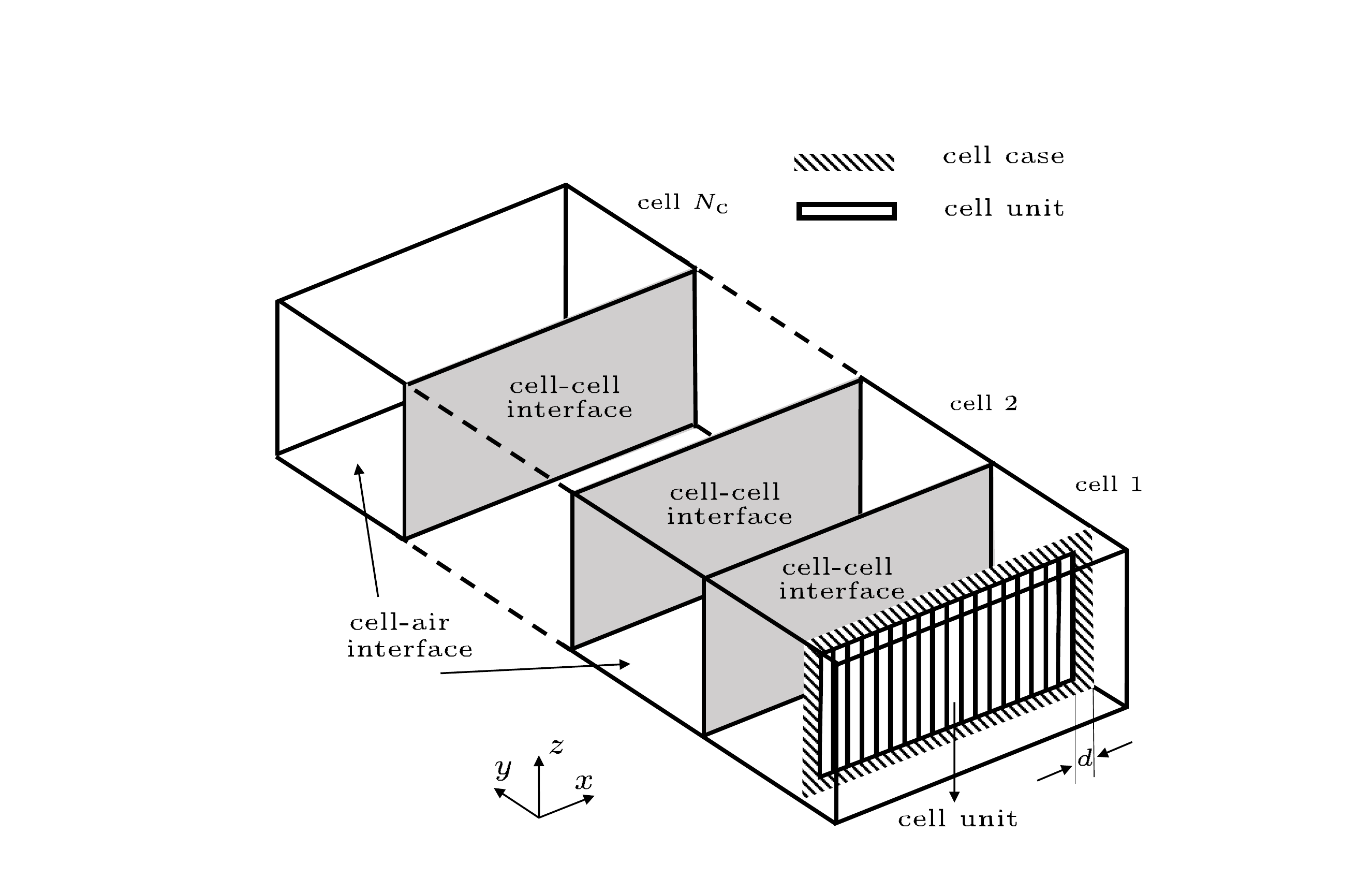}
\centering
\caption{Schematic diagram of a LiB pack.}
\label{Fig:LiB-Pack-Schematic}
\end{figure}

To proceed further, a brief review of the heat transfer mechanisms is offered. It is known that there are three ways for heat transfer from one place to another: conduction, convection and radiation. Conduction happens with in a substance or between substances in direct contact, caused by collision between atoms. Convection generally refers to the heat transfer with a fluid that moves between areas with temperature difference. Radiation is the energy emission by objects at nonzero temperature in the form of electromagnetic waves. An interested reader is referred to~\cite{bergman2006fundamentals} for more details.

Consider the thermal dynamics in the core region first. Here, convection and radiation can be ignored since electromagnetic waves can hardly transmit through the cell and the liquid electrolyte is of limited mobility. Consequently,
heat transfer within the core region is dominated by conduction, which can be expressed by the three-dimensional heat diffusion equation

\begin{align}\label{thermal-model}
\rho_{\rm co} c_{\rm co} \frac{\partial{T}}{\partial{t}}=\lambda_{\rm co} \nabla^2 T+q,
\end{align}
where $\rho_{\rm co}$, $c_{\rm co}$ and $\lambda_{\rm co}$ are the mass density, specific heat and thermal conductivity of the core region, respectively. In addition, $T$ and $q$, respectively, denote the core region's temperature in kelvins and heat generation density. A general characterization of $q$, as discussed in~\cite{al1999thermal,chen2006thermal,bandhauer2011critical}, is offered by
\begin{align}\label{heat-generation}
q=\frac{I}{V_{\rm co}}\left[\left(U_{\rm ocv}-U_{\rm t}\right) -T\frac{\textrm{d} U_{\rm ocv}}{\textrm{d} T} \right],
\end{align}
where $V_{\rm co}$, $I$, $U_{\rm ocv}$, $U_{\rm t}$ and $\textrm{d} U_{\rm ocv}/ \textrm{d} T$ denote the total volume of the core region, the current through pack (positive for discharge, negative for charge), the open-circuit voltage, the terminal voltage and the entropic heat coefficient, respectively. Here, the heat is assumed to be generated uniformly across the core region. The first term on the right-hand side of~\eqref{heat-generation} is the irreversible heating, and the second term is the reversible entropic heating from electrochemical reactions.

Similarly, conduction makes up the dominant part of the heat transfer in the case region, that is,
\begin{align}\label{thermal-model-case}
\rho_{\rm ca} c_{\rm ca} \frac{\partial{T}}{\partial{t}}=\lambda_{\rm ca} \nabla^{2} T,
\end{align}
where $\rho_{\rm ca}$, $c_{\rm ca}$ and $\lambda_{\rm ca}$ are the mass density, specific heat and thermal conductivity of the case, respectively. 

Next, consider the boundaries. To begin with, heat transfer on the core-case interface is mainly due to conduction. Assuming that core region and cell case are in perfect contact, the temperature and heat flux can be considered continuous on the interface~\cite{shiming2006heat}. Hence, the boundary conditions on the core-case interface are given by
\begin{equation}\label{boundary-core-case}
\begin{aligned}
T|_{\rm core} &= T|_{\rm case}, \ \lambda_{\rm co}\frac{\partial T}{\partial n}\bigg|_{\rm core} = \lambda_{\rm ca}\frac{\partial T}{\partial n}\bigg|_{\rm case},
\end{aligned}
\end{equation}
where $n$ is the normal direction of the core-case interface. The continuity of temperature and heat flux at a boundary also holds for the cell-cell interface with assumption that LiB cells are in close contact. Take cells 1 and 2 in Figure~\ref{Fig:LiB-Pack-Schematic} as an example. On their cell-cell interface,
\begin{equation}\label{boundary-cell-cell}
\begin{aligned}
T|_{\rm cell \ 1} &= T|_{\rm cell \ 2}, \
\lambda_{\rm ca}\frac{\partial T}{\partial n}\bigg|_{\rm cell \ 1} = \lambda_{\rm ca}\frac{\partial T}{\partial n}\bigg|_{\rm cell \ 2},
\end{aligned}
\end{equation}
where $n$ is the normal direction of the interface of cells 1 and 2. Different from the regions and interfaces above, the cell-air interface will see heat transfer due to all of conduction, convection and radiation. Therefore, the energy balance on this boundary is
\begin{align}\label{boundary1}
\lambda_{\rm ca}\frac{\partial T}{\partial n}\bigg|_{\rm cell} =q_{\rm conv}+q_{\rm r},
\end{align}
where $q_{\rm conv}$ and $q_{\rm r}$ represent the convective heat flux and the radiative heat flux on the cell-air interface, respectively. Here, $q_{\rm conv}$ and $q_{\rm r}$ are given by
\begin{align}\label{convection}
q_{\rm conv}&=h_{\rm conv} \left(T-T_{\rm air}\right),\\ \label{radiation}
q_{\rm r}&=\varepsilon \sigma (T^4 - T_{\rm air}^4 ),
\end{align}
where $h_{\rm conv}$, $\varepsilon$, $\sigma$ and $T_{\rm air}$ are the convective heat transfer coefficient, the emissivity, the Stefan-Boltzmann constant and the ambient air temperature, respectively. Note that if $|T-T_{\rm air}|/T_{\rm air}\ll 1$ as is often the case of LiB operation (both $T$ and $T_{\rm air}$ in kelvins),~\eqref{radiation} can be linearized around $T_{\rm air}$ as 
\begin{align}\label{radiation-linear}
q_{\rm r} = 4 \varepsilon \sigma T_{\rm air}^3 (T-T_{\rm air})=h_{\rm r}(T-T_{\rm air}),
\end{align}
where $h_{\rm r}$ is the radiative heat transfer coefficient. Combining \eqref{convection} and \eqref{radiation-linear}, the energy balance \eqref{boundary1} can be rewritten as
\begin{align}\label{boundary}
\lambda_{\rm ca}\frac{\partial T}{\partial n}\bigg|_{\rm cell}=h(T-T_{\rm air}),
\end{align}
where $h$ is the combined convective-radiative heat transfer coefficient.

Summarizing~\eqref{thermal-model}-\eqref{boundary}, one will obtain a complete thermal model for the considered LiB pack. This model builds on an awareness of the spatial temperature distribution. According to~\eqref{thermal-model}-\eqref{heat-generation}, heat will be produced within the core region when a current flows through the pack and transferred across the region by conduction. Conduction will also enable the propagation of heat within the case region, which is shown in~\eqref{thermal-model-case}. The boundary conditions at the core-case and cell-cell interfaces can be determined as in~\eqref{boundary-core-case}-\eqref{boundary-cell-cell} on the reasonable assumption of continuous temperature and heat flux. Meanwhile, heat will travel from the cell surface to the air driven by a mix of conduction, convection and radiation, as shown in~\eqref{boundary1}. The radiation effect at the cell-air interface is further linearized to simplify the model, as allowed by some mild conditions that can be easily met in a LiB pack's operation. This would lead to~\eqref{boundary}. Finally, the initial condition is also specified. In spite of the PDE-based representation, the obtained model still has constrained overall complexity advantageous for computation. In the next section, it will be further converted to ODE-based state-space form for the purpose of estimation design.

{\color{blue}
\begin{remark}\label{Remark:cooling-device}
{\rm (Extensions of the thermal model).} The thermal model above is developed in a basic battery pack setting and able to capture the most critical heat transfer phenomena underlying a pack's thermal behavior. Meanwhile, it can be readily extended to deal with more sophisticated settings. 1) Extension to a battery pack with a cooling system. The cooling effects can be accounted for in two ways. First, as suggested in~\cite{kim2014estimation}, one can regard the cooling system as the boundaries of the battery pack's thermal model and thus modify the boundary conditions accordingly. Second, one can develop a separate heat transfer model for the cooling system and determine its interaction with the pack's model. The two models can be combined to offer a complete description of the battery pack under cooling conditions. This idea is exploited in~\cite{mahamud2011reciprocating,lin2011parameterization,lin2014temperature}. 2) Extension to nonuniform heat generation. As is shown in~\eqref{heat-generation}, the heat generation is assumed to be even across a cell. The model built on this assumption can be modified to cope with spatially nonuniform heating if the gradient distribution of the potential and current density is captured. The literature includes some studies in this regard, see~\cite{bandhauer2011critical} and the references therein. In addition, the heat produced by electrical connections can also be included into the thermal model, which supposedly results from the passage of an electrical current through a resistance~\cite{kim2008effect}. 3) Extension to heterogeneous cells. While the above considers identical cells, cells of the same type but of different state or aging level can be dealt with by changing the model parameters. Further, if cells of different electrochemistries are used in an extreme case, one can first build separate models for each cell type and then couple them using the same heat transfer principles to obtain a pack-level model. It is noteworthy that, though based on the basic model, the methodology to be proposed next is still applicable to these extended models for thermal field reconstruction.
\hfill$\bullet$
\end{remark}
}

\section{Discrete State-Space Equation}\label{Sec:Model-Conversion}

In this section, the PDE-based model in Section~\ref{Sec:Modeling} is discretized in space and time to derive the state-space model using the finite difference formulation based on energy balance.

\subsection{Discretization}\label{Sec:discretization-subsection}

The finite difference formulation is applied to fulfill the discretization. Consider the LiB pack comprising $N_c$ cells shown in Figure~\ref{Fig:LiB-Pack-Schematic}. It can be subdivided into a large number of volume elements, giving rise to a three-dimensional grid with many nodes, see Figure~\ref{Fig:grid-of-pack}. Specifically, a LiB cell, with dimensions of $L_x \times L_y \times L_z$, is subdivided by a grid with $(m\times n\times p)$ nodes. Here, two adjacent cells share the $mp$ nodes on the contact interface between them. Then, each node can be labeled by its $xyz$-coordinates, i.e., $(i,j,k)$, which ranges from $(1,1,1)$ to $(m,nN_{\mathrm{c}}-N_{\mathrm{c}}+1,p)$. A node is linked with a control volume. Within this volume, the temperature is considered uniform and assigned as the temperature of the node.
The control volume of a node in the core region is uniform with a size of $(a\times b\times c)$, where $a$, $b$ and $c$ are the volume's edge length in $x$, $y$ and $z$ directions, respectively; yet the control volume of a node in the case region varies subject to its specific location. Since the case thickness is quite slim relative to the cell's dimensions, the temperature change in the thickness direction then will be minor. Thus for the node in the case, we let its edge parallel to the case's thickness direction be larger than the case's thickness length $d$, e.g., the control volume of (2,1,2) has a size of $(a\times (b/2+d)\times c)$. As such, $
L_x = (m-1)a+2d$,
$L_y = (n-1)b+2d$, and
$L_z = (p-1)c+2d$.

\begin{figure}[t]
\centering
\includegraphics[trim = {60mm 20mm 50mm 30mm}, clip, width=0.4\textwidth]{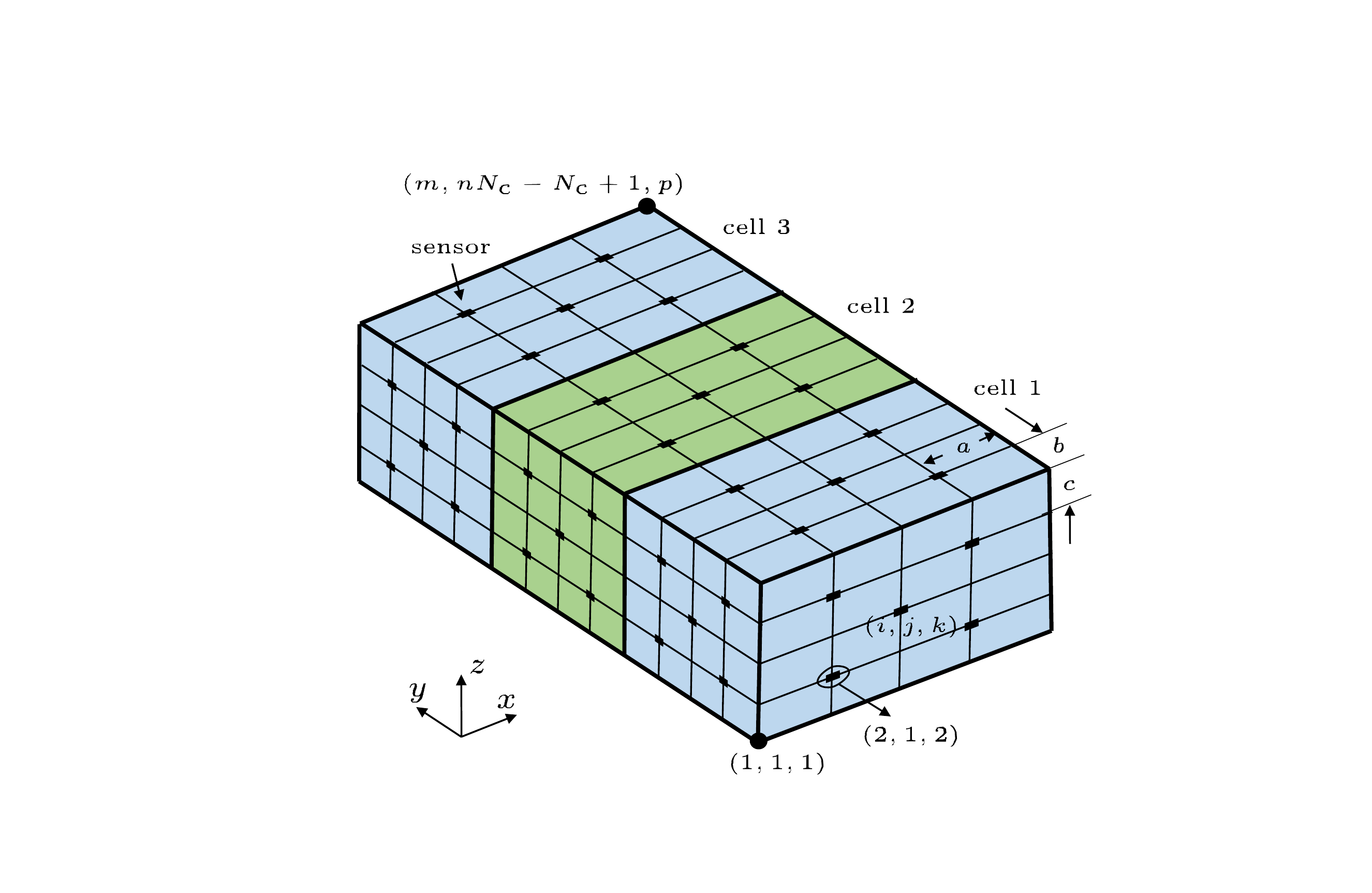}
\centering
\caption{Schematic diagram of the nodes in a LiB pack comprising three cells, namely, $N_{\rm c}=3$.}
\label{Fig:grid-of-pack}
\end{figure}

\begin{figure}[t]
\centering
\includegraphics[trim = {80mm 35mm 65mm 25mm}, clip, width=0.3\textwidth]{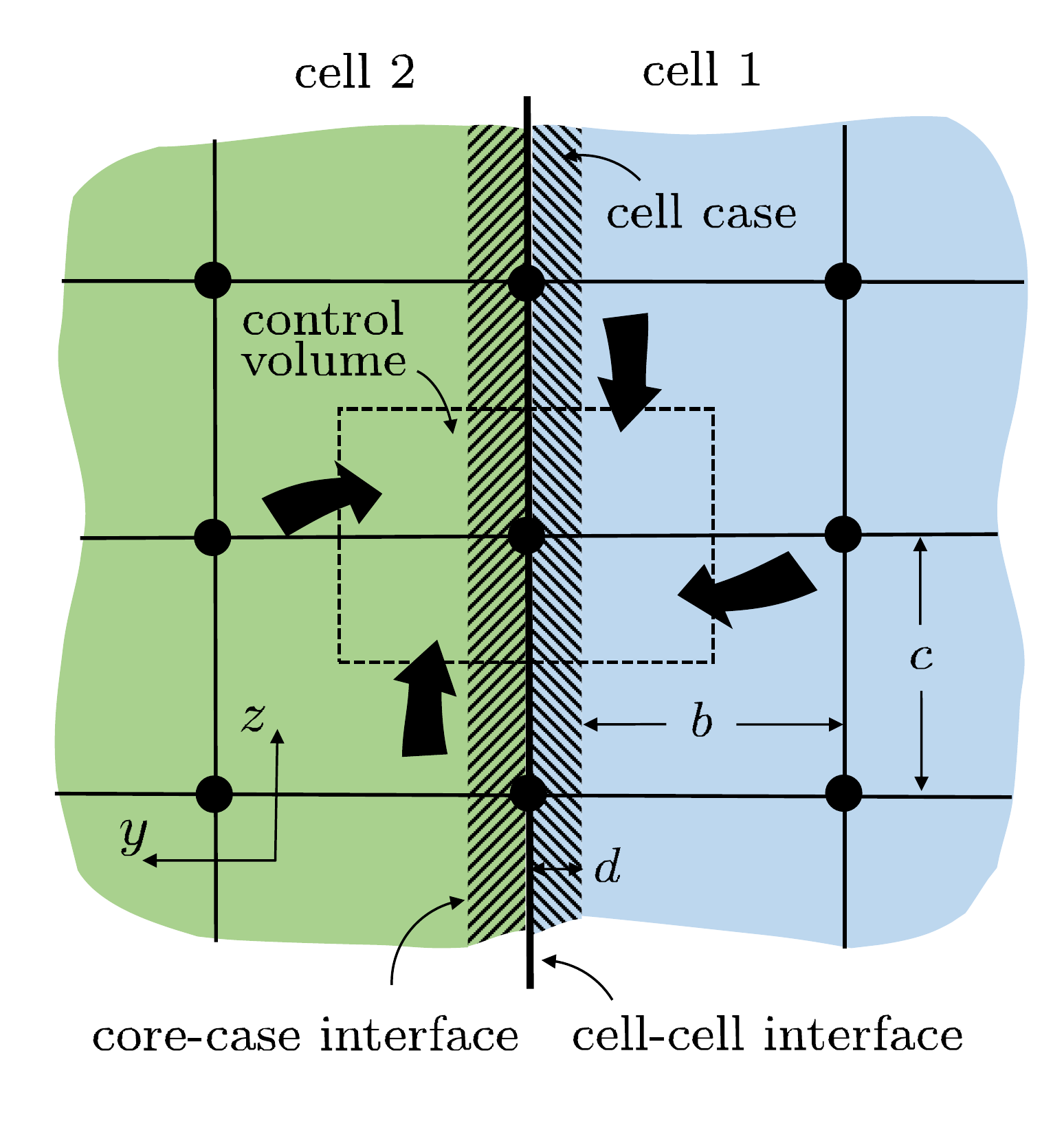}
\centering
\caption{The $y$-$z$ cross-section of the control volume of a node on the cell-cell interface.}
\label{Fig:interface-grid}
\end{figure}

In this paper, the finite-difference equation at each node is developed by the \textit{energy balance approach} \cite{bergman2006fundamentals}. That is, the energy conservation law is applied to each volume. Here the direction of heat transfer on each volume's surfaces is assumed to be toward the node, see Figure~\ref{Fig:interface-grid}. 
For a node in the core region, e.g., (2,2,2) of the grid in Figure~\ref{Fig:grid-of-pack}, it receives the conductive heat flow from its surrounding six nodes. The finite-difference equation can then be expressed as 
\begin{align} \label{core-region}
\begin{split}
\rho & c_p abc \frac{T_{i,j,k}^{t+\Delta t}-T_{i,j,k}^{t}}{\Delta t}
=\lambda_{{\rm co},x} \frac{T_{i+1,j,k}^{t}-T_{i,j,k}^{t}}{a}bc \\ &
+\lambda_{{\rm co},x} \frac{T_{i-1,j,k}^{t}-T_{i,j,k}^{t}}{a}bc 
+\lambda_{{\rm co},y} \frac{T_{i,j+1,k}^{t}-T_{i,j,k}^{t}}{b}ac \\ &
+\lambda_{{\rm co},y} \frac{T_{i,j-1,k}^{t}-T_{i,j,k}^{t}}{b}ac 
+\lambda_{{\rm co},z} \frac{T_{i,j,k+1}^{t}-T_{i,j,k}^{t}}{c}ab \\ &
+\lambda_{{\rm co},z} \frac{T_{i,j,k-1}^{t}-T_{i,j,k}^{t}}{c}ab
+\frac{abc(U_{\rm ocv}-U_{\rm t}-\frac{\textrm{d} U_{\rm ocv}}{\textrm{d} T}T_{i,j,k}^{t})}{\rm{V}_{\rm co}}I,
\end{split}
\end{align}
{\color{blue}where $\rho$, $c_p$ and $\Delta t$ are the mass density, specific heat capacity of the cell and time step, respectively.} Here, the thermophysical parameters, (e.g., $\rho$, $c_p$ and $\lambda_{\rm co}$), are time invariant, and $\mathrm{d}U_{\rm ocv}/{\mathrm{d} T}$ is also assumed to be independent of time within the operating range of the LiB pack~\cite{chen2005thermal}. Therefore, \eqref{core-region} can be regarded as a linear equation. The product of $\rho$ and $c_p$ is averaged based on the volume of the core region $V_{\rm co}$ and the case $V_{\rm ca}$ for simplification, namely,
\begin{align}\label{density}
\rho c_p = \frac{\rho_{\rm co} c_{\rm co} V_{\rm co}+\rho_{\rm ca} c_{\rm ca} V_{\rm ca}}{V_{\rm co}+V_{\rm ca}}.
\end{align}
The strategy in the calculation of $\rho_{\rm co} c_{\rm co}$, $\rho_{\rm ca} c_{\rm ca}$ and the three-dimensional thermal conductivity $\lambda_{\rm co}$ is the same with that used in~\cite{chen2005thermal} and thus omitted here.

Next, consider a node on the cell-cell boundary interface as shown in Figure~\ref{Fig:interface-grid}. Its control volume in the $y$ direction is $2d+b$. With properly selected edge length in $x$ and $z$ directions, the control volume of this node can be regarded as isothermal. Since the core region and the case region are made up of different materials, the thermal conduction length between the node and its adjacent node in the $y$ direction is $b$ instead of $b+d$. The finite-difference equation \cite{cengel2015heat} for this boundary is thus given by
\begin{align} \label{interface}
\begin{split}
& \rho c_p ac(b+2d) \frac{T_{i,j,k}^{t+\Delta t}-T_{i,j,k}^{t}}{\Delta t}
=\lambda_{{\rm co},x} \frac{T_{i+1,j,k}^{t}-T_{i,j,k}^{t}}{a}bc \\ &
+\lambda_{{\rm ca},x} \frac{T_{i+1,j,k}^{t}-T_{i,j,k}^{t}}{a}2cd 
+\lambda_{{\rm co},x} \frac{T_{i-1,j,k}^{t}-T_{i,j,k}^{t}}{a}bc \\ &
+\lambda_{{\rm ca},x} \frac{T_{i-1,j,k}^{t}-T_{i,j,k}^{t}}{a}2cd 
+\lambda_{{\rm co},y} \frac{T_{i,j+1,k}^{t}-T_{i,j,k}^{t}}{b}ac \\ &
+\lambda_{{\rm co},y} \frac{T_{i,j-1,k}^{t}-T_{i,j,k}^{t}}{b}ac 
+\lambda_{{\rm co},z} \frac{T_{i,j,k+1}^{t}-T_{i,j,k}^{t}}{c}ab \\ &
+\lambda_{{\rm ca},z} \frac{T_{i,j,k+1}^{t}-T_{i,j,k}^{t}}{c}2ad 
+\lambda_{{\rm co},z} \frac{T_{i,j,k-1}^{t}-T_{i,j,k}^{t}}{c}ab \\ &
+\lambda_{{\rm ca},z} \frac{T_{i,j,k-1}^{t}-T_{i,j,k}^{t}}{c}2ad
+\frac{abc(U_{\rm ocv}-U_{\rm t}-\frac{\textrm{d} U_{\rm ocv}}{\textrm{d} T}T_{i,j,k}^{t})}{\rm{V}_{\rm co}}I, 
\end{split}
\end{align}

Then, consider a node on the cell-air boundary interface of the pack. Take the node (2,1,2) in Figure~\ref{Fig:grid-of-pack} as an example. It receives both conductive heat transfer from its surrounding nodes and combined convective-radiative heat transfer from the ambient air. Similar to the nodes on the cell-cell interface, the edge length of the control volume of node (2,1,2) should be larger than $d$ in the $y$ direction, and the effective thermal conduction length between the node and its adjacent node in the $y$ direction would also be equal to $b$. Hence, the finite-difference equation will be 
\begin{align} \label{surface}
\begin{split}
\rho & c_p ac\left(\frac{b}{2}+d\right) \frac{T_{i,j,k}^{t+\Delta t}-T_{i,j,k}^{t}}{\Delta t}
=\lambda_{{\rm co},x} \frac{T_{i+1,j,k}^{t}-T_{i,j,k}^{t}}{a}\frac{b}{2}c \\ & 
+\lambda_{{\rm ca},x} \frac{T_{i+1,j,k}^{t}-T_{i,j,k}^{t}}{a}cd 
+\lambda_{{\rm co},x} \frac{T_{i-1,j,k}^{t}-T_{i,j,k}^{t}}{a}\frac{b}{2}c \\ &
+\lambda_{{\rm ca},x} \frac{T_{i-1,j,k}^{t}-T_{i,j,k}^{t}}{a}cd 
+\lambda_{{\rm co},y} \frac{T_{i,j+1,k}^{t}-T_{i,j,k}^{t}}{b}ac \\ &
+\lambda_{{\rm co},z} \frac{T_{i,j,k+1}^{t}-T_{i,j,k}^{t}}{c}a\frac{b}{2} 
+\lambda_{{\rm ca},z} \frac{T_{i,j,k+1}^{t}-T_{i,j,k}^{t}}{c}ad \\ &
+\lambda_{{\rm co},z} \frac{T_{i,j,k-1}^{t}-T_{i,j,k}^{t}}{c}a\frac{b}{2} 
+\lambda_{{\rm ca},z} \frac{T_{i,j,k-1}^{t}-T_{i,j,k}^{t}}{c}ad \\ &
+ach(T_{\rm air} - T_{i,j,k}^{t}) 
+\frac{abc(U_{\rm ocv}-U_{\rm t}-\frac{\textrm{d} U_{\rm ocv}}{\textrm{d} T}T_{i,j,k}^k)}{2\rm{V}_{\rm co}}I.
\end{split}
\end{align}

Thus far, the PDE-based thermal model has been systematically discretized in space and time. The obtained equations present a linear ODE system that is high-dimensional but more amenable to estimation design and implementation. %

\subsection{State-Space Model Development}\label{Sec:state-subsection}

Reorganizing the ODEs~\eqref{core-region}-\eqref{surface}, one can derive a high-dimensional state-space representation of the following general form:
\begin{align}\label{state-equation}
\bm{x}_{k+1} = \bm{F}_k \bm{x}_{k} + \bm{G}_k\bm{u}_k,
\end{align}
where $\bm{x}$ is the state vector summarizing the temperature at all the nodes, $\bm{u}$ the input vector based on the ambient temperature $T_{\rm air}$ and $I$ the current. 
Specifically,
\begin{align*}
\bm{x} =
\left[ \begin{array}{c}
\bm{x}^{(1)} \\
\mathbf{\vdots} \\
\bm{x}^{(N_{\rm c})} 
\end{array} \right] 
\in \mathbb{R}^{N\times 1},
\end{align*}
where
\begin{align*}
\bm{x}^{(l)} =
\left[ \begin{matrix}
T_{1,1,1}^{(l)} &
\cdots &
T_{m,n,p}^{(l)} 
\end{matrix} \right]^\T 
\in \mathbb{R}^{N_l\times 1},
\end{align*}
Here, for ease of the decomposition strategy design toward distributed estimation, the global state vector $\bm x$ is established as an aggregation of local state vector $\bm x^{(l)}$ of each cell. To accommodate this representation, the pack-based numbering method introduced in Section~\ref{Sec:discretization-subsection} is modified accordingly to be cell-based---for each cell, the nodes are numbered from $(1,1,1)$ to $(m,n,p)$. Although the length of the vector $\bm x$ is increased from $N-mp(N_{\rm c}-1)$ to $N$ as a result of the nodes in the cell-cell interface accounted for separately, this will allow each cell to be treated as a subsystem and pave the way for distributed estimation algorithm design, as will be shown later. Besides, for notational convenience, the discrete time index will be denoted by $k$ in the remainder of the paper in replacement of $t$.
The input $\bm{u}$ is given by
\begin{align*}
\bm{u} =
\left[ \begin{matrix}
T_{\rm air} &
I 
\end{matrix} \right]^\T \in \mathbb{R}^{2\times 1},
\end{align*}
which sums up the ambient temperature and current in the role of driving the system.
Furthermore, the matrices $\bm{F}$ and $\bm{G}$ can be expressed in the block form:
\begin{align*}
\bm{F}&=
\left[ \begin{array}{cccc}
\bm{F}_{11} & \bm{F}_{12} & &\mathbf{0} \\
\bm{F}_{21} & \bm{F}_{22} & {\ddots} & \\
& {\ddots} & {\ddots} &{\ddots} \\
\mathbf{0} & & {\ddots} & \bm{F}_{N_{\rm c},N_{\rm c}} 
\end{array} \right] \in \mathbb{R}^{N\times N},\\
\bm{G}&=
\left[ \begin{array}{c}
\bm{G}_{1} \\
\vdots \\
\bm{G}_{N_{\rm c}} 
\end{array} \right] \in \mathbb{R}^{N\times 2},
\end{align*}
where the block entries can be readily determined from~\eqref{core-region}-\eqref{surface}.
It is interesting to note that the matrix $\bm{F}$ is not only tridiagonal but also diagonally dominant and sparse. This is mainly because each cell will only exchange heat with its adjacent cells and adjacent cells only share a limited number of nodes.

Let cell $l$ in the pack be equipped with $M_l$ thermocouples to measure the temperature. The measurement equation can then be expressed as 
\begin{align}\label{measurement-equation-1}
\bm{y}_{k}^{(l)} = \bm{H}_l \bm{x}^{(l)}_k,
\end{align}
where $\bm{y}^{(l)} \in \mathbb{R}^{M_l} $ and $\bm{H}_l \in \mathbb{R}^{M_l \times N_l}$ are the temperature measurement and measurement matrix, respectively. For $\bm{H}$, the entries corresponding to the nodes directly measured will be set equal to 1, and all the other entries zero. Aggregating all the measurements together, the measurement equation for the entire pack is then given by
\begin{align}\label{measurement-equation-2}
\bm{y}_{k} = \bm{H} \bm{x}_k,
\end{align}
where 
\begin{align*}
\bm{y} &=
\left[ \begin{matrix}
\bm{y}^{(1)} \\
{\vdots} \\
\bm{y}^{(N_{\rm c})} 
\end{matrix} \right] \in \mathbb{R}^{M\times 1}, \\
\bm{H} &= \left[ \begin{matrix}
\bm{H}_1 & & \\
& {\ddots} & \\
& & \bm{H}_{N_{\rm c}} 
\end{matrix} \right] \in \mathbb{R}^{M\times N}, 
\end{align*}
where $M=N_{\rm c} M_l$. 

From above,~\eqref{state-equation} and~\eqref{measurement-equation-2} form the state-space model, which characterizes the propagation and measurement of a LiB pack's thermal dynamics based on a group of ODEs. With this model, one will be in a good position to conduct state estimation toward reconstructing the temperature field.

\section{Kalman Filter Estimation}\label{Sec:Estimation}

The state-space thermal model developed in Section~\ref{Sec:Model-Conversion} describes the dynamic thermal behavior of a LiB pack in charging/discharging. With this model, optimal estimation can be employed to enable reconstruction of the temperature field from the temperature measurement data. In this work, the KF will be exploited as the solution tool, which has achieved proven success in a wide range of engineering fields. However, the standard centralized KF will not fit with this application, because of a fine characterization of the LiB pack's thermal dynamics will imply a large size of the state space and cause heavy computation. This thus motivates the use of a distributed approach to build computational efficiency. In the following, the centralized KF will be introduced first, and then a distributed version presented and analyzed in detail.

\subsection{Centralized Kalman Filtering}\label{Sec:CKF}

Consider the state-space equations~\eqref{state-equation} and~\eqref{measurement-equation-2} and for convenience, replicate them with noise terms added as follows:
\begin{equation}\label{new-system}
\left\{
\begin{aligned}
\bm{x}_{k+1} &= \bm{F}_k \bm{x}_k + \bm{G}_k \bm{u}_k + \bm{w}_k,\\
\bm{y}_k &= \bm{H} \bm{x}_k + \bm{v}_k.
\end{aligned}
\right.
\end{equation}
Here, the vectors with compatible dimensions, $\bm w_k$ and $\bm v_k$, are are added to account for the process noise and measurement noise that exist in the thermal dynamic processes of a LiB pack. They are assumed to be zero-mean Gaussian white noises with covariances of $\bm{Q}\geq 0$ and $\bm{R}>0$, respectively. Suppose the initial state $\bm{x}_0$ is a Gaussian random vector with mean $\hat{\bm{x}}_{0|0}$ and covariance $\bm{P}_{0|0}$. Note that $\hat{\bm x}_{0|0}$ indeed makes up the initial guess of $\bm{x}_0$, and that $\bm{P}_{0|0}$ represents the estimation error covariance. Then
application of the standard CKF to~\eqref{new-system} can be performed at each future time instant. This procedure consists of two steps, {\em prediction} and {\em update}. The one-step-forward prediction yields the estimate of $\bm{x}_k$, denoted as $\hat{\bm{x}}_{k|k-1}$, using the measurements collected up to time $k-1$. Then upon the arrival of $\bm{y}_k$, $\hat{\bm{x}}_{k|k-1}$ will be updated to $\hat{\bm{x}}_{k|k}$ leveraging the information conveyed by $\bm{y}_k$ about $\bm{x}_k$. In the meantime, the estimation error covariances associated with both estimates are computed accordingly to quantify the uncertainties of the obtained estimates.

When the state estimate $\hat{\bm{x}}_{k-1|k-1}$ is generated, the one-step-forward prediction can be made through
\begin{align}\label{CKF-prediction1}
\hat{\bm{x}}_{k|k-1} &= \bm{F}_{k-1} \hat{\bm{x}}_{k-1|k-1}+ \bm{G}_{k-1}\bm{u}_{k-1},\\ \label{CKF-prediction2}
\bm{P}_{k|k-1} & = \bm{F}_{k-1} \bm{P}_{k-1|k-1} \bm{F}_{k-1}^{\top}+ \bm{Q}.
\end{align}
After $\hat x_{k|k-1}$ is produced, it will be of next interest to investigate the updated state estimate. When the new measurement $\bm{y}_k$ becomes available, the update step can be performed as follows to correct the prediction:
\begin{align}\label{CKF-filter1}
\hat{\bm{x}}_{k|k} &= \hat{\bm{x}}_{k|k-1} + \bm{K}_{k} (\bm{y}_{k}-\bm{H} \hat{\bm{x}}_{k|k-1} ) , \\ \label{CKF-filter2}
\bm{K}_{k} &= \bm{P}_{k|k-1} \bm{H}^{\top} (\bm{H} \bm{P}_{k|k-1} \bm{H}^{\top}+ \bm{R} ) ^{-1}, \\ \label{CKF-filter3}
\bm{P}_{k|k}& = \bm{P}_{k|k-1} - \bm{K}_{k} \bm{H} \bm{P}_{k|k-1} .
\end{align}
The CKF will execute the above steps recursively over time to generate the state estimate at each time instant. The CKF is the optimal among all filters if the aforementioned Gaussian assumptions are satisfied and optimal among all linear filters in the non-Gaussian case. Such an optimality establishes the foundation for the practical utility of the CKF approach. The CKF, however, is not well suited to estimate the temperature for a LiB pack, because the high-dimensional thermal model of the LiB pack will imply considerable computation (detailed computational complexity analysis will be given in Section~\ref{Sec:Complexity-Analysis}). Hence, the CKF will be distributed next to bring down the computational cost.

\subsection{Distributed Kalman Filtering}\label{Sec:DKF}

Rather than estimate the global state in a centralized manner, the DKF will consider the pack system as a combination of multiple cell-based subsystems and run a series of local KFs in parallel, each one corresponding to a cell. Because of the mutual influence between the cells and their thermal behavior, the local KFs will exchange information according to the existing communication topology to accomplish the estimation. The local estimates, when collected and put together, will comprise a complete picture of the entire pack's temperature field.

Consider a LiB pack composed of $N_{\rm c}$ cells wired in series, which are numbered in order from 1 to $N_{\rm c}$. A three-cell pack is shown in Figure~\ref{Fig:Cell-Network} as an example. For cell $l$ and $i$, they are said to be neighbors if they are adjacent. The
neighborhood of $l$, $\calN_l$, is defined as the set of its neighbor cells, and in this setting, $\calN_l = \{l-1, l+1\}$. Here, it is assumed that cell $l$ stores the structural information of its neighbors, specifically, $\bm{F}_{li}$ for $i\in \calN_l$. To minimize the communication cost, information exchange protocol is only enforced between neighbors. That is, cell $l$ only communicates with $\calN_l$.

\begin{figure}[t]
\centering
\includegraphics[trim = {0mm 40mm 10mm 10mm}, clip, width=0.5\textwidth]{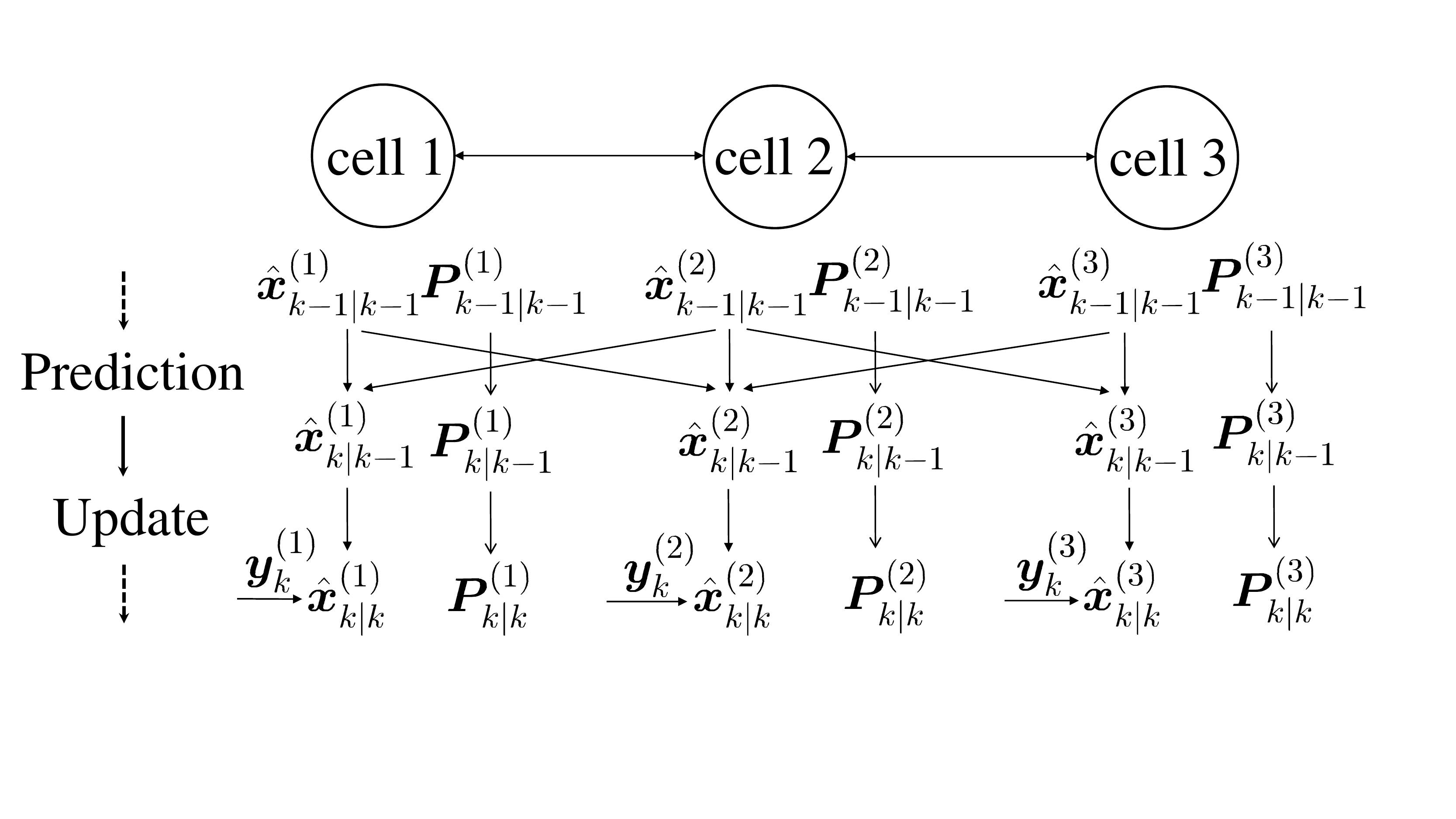}
\centering
\caption{A schematic of the DKF for a three-cell pack.}
\label{Fig:Cell-Network}
\end{figure}

As such, the state-space equation for subsystem $S_l$ can be written as
\begin{equation}\label{Local-Subsystem}
\left\{
\begin{aligned}
\bm{x}_{k+1}^{(l)}&= \bm{F}_{ll,k} \bm{x}_{k}^{(l)} + \sum_{i\in \mathcal{N}_l}\bm{F}_{li,k}\bm{x}_k^{(i)}+\bm{G}_{l,k} \bm{u}_{k}+\bm{w}_{k}^{(l)} , \\ 
\bm{y}_{k}^{(l)} &= \bm{H}_l \bm{x}^{(l)}_k +\bm{v}^{(l)}_k,
\end{aligned}
\right.
\end{equation}
where cell-wise decomposition is also applied to $\bm w_k$ and $\bm v_k$. It is noted that the evolution of cell $l$'s state is not only self-driven but also affected by the neighboring cells, as a result of the cell-to-cell coupling of the thermal dynamics. All cells yet share the same $\bm{u}_k$ because the serial connection implies the same charging/discharging current across the circuit. In addition, each cell is only aware of its own temperature measurements. 

For the above cell $l$-based subsystem, the KF approach can be adjusted for local state estimation. This can be attained by applying the prediction-update procedure in analogy to~\eqref{CKF-prediction1}-\eqref{CKF-filter3}. Specifically, the prediction can be given by
\begin{align}\nonumber
\hat{\bm{x}}_{k|k-1}^{(l)}= &\bm{F}_{ll,k-1} \hat{\bm{x}}_{k-1|k-1}^{(l)} + \sum_{i\in \mathcal{N}_l}\bm{F}_{li,k-1}\hat{\bm{x}}_{k-1|k-1}^{(i)}\\  \label{DKF-prediction1}
&+\bm{G}_{l,k-1} \bm{u}_{k-1} , \\ \label{DKF-prediction2}
\bm{P}^{(l)}_{k|k-1} = &\bm{F}_{ll,k-1} \bm{P}_{k-1|k-1}^{(l)} \bm{F}^{\top}_{ll,k-1} + \bm{Q}_l.
\end{align}
Here, cell $l$'s state prediction, $\hat{\bm{x}}_{k|k-1}^{(l)}$, depends on not only its own but also its neighbors' state estimates from the previous time instant. All cells are still driven by the same input $\bm{u}_k$ because cells in serial connection are subjected to the same current.
On its arrival, $\bm{y}_k^{(l)}$ can be used to update $\hat{\bm{x}}_{k|k-1}^{(l)}$ as follows:
\begin{align}\label{DKF-filter1}
\hat{\bm{x}}^{(l)}_{k|k} &= \hat{\bm{x}}^{(l)}_{k|k-1} + \bm{K}_{k}^{(l)} (\bm{y}^{(l)}_{k}-\bm{H}_l \hat{\bm{x}}_{k|k-1}^{(l)} ) , \\ \label{DKF-filter2}
\bm{K}_{k}^{(l)} &= \bm{P}_{k|k-1}^{(l)} \bm{H}^{\top}_l (\bm{H}_l \bm{P}_{k|k-1}^{(l)} \bm{H}^{\top}_l+ \bm{R}_l)^{-1}, \\ \label{DKF-filter3}
\bm{P}_{k|k}^{(l)} &= \bm{P}_{k|k-1}^{(l)} - \bm{K}_{k}^{(l)} \bm{H}_l \bm{P}_{k|k-1}^{(l)}.
\end{align}
Note that no information exchange with neighbors is required in the update step. The two steps shown in~\eqref{DKF-prediction1}-\eqref{DKF-filter3} then constitute the DKF algorithm. Running the DKF in parallel for each cell, the local temperature field will be estimated through time, and all the local estimation results when combined will provide the full view of the pack's temperature field.

It is noteworthy that the above DKF algorithm, compared with the CKF, involves approximation for two reasons. First, each cell only has a local rather than global knowledge of the system's dynamic behavior, making certain information loss inherent in each local DKF. To see this, the nominal prediction error covariance $\bm{P}^{(l)}_{k|k-1}$, differing from the CKF, evolves only from its predecessor without fusion of the counterparts of the other cells, as shown in~\eqref{DKF-prediction2}. This shows $\bm{P}^{(l)}_{k|k-1}$ is only approximate to the true prediction error covariance. Second, part of the approximation is made to reduce the communication and computation costs. Looking at~\eqref{DKF-prediction2} again, one can see that a local cell does not consider its neighbors in its forward propagation of its prediction error covariance. This will obviate the need for the exchange of the estimation error covariance between neighboring cells and further, the local computational effort. However, such an approximation will not seriously compromise the estimation accuracy. Since $\bm{F}$ is diagonally dominant and $\bm{H}$ block-diagonal due to the pack's serial connection architecture, $\bm{P}$ will also be diagonally dominant. The self-propagation of the local estimation error covariance, as a result, will not bring much loss of estimation accuracy. The above DKF algorithm, to our knowledge, is the most computationally fast among its kind and highly suitable for the considered LiB pack thermal monitoring problem.

\begin{table*}\centering
\caption{Arithmetic operation requirements of the CKF, DKF and SS-DKF algorithms.}
\begin{tabular}{cccc}
\toprule%
Algorithm & Number of multiplications & Number of additions & Complexity \\
\midrule
CKF & $3N^3+2N^2M+2NM^2+M^3+N^2+2NM+2N$ & $3N^3+2N^2M+2NM^2+M^3-N^2+N$ & $O(N^3)$\\
DKF & $N_{\rm c}\left(3N_l^3+2N_l^2M_l+2N_lM_l^2+M_l^3+3N_l^2+2N_lM_l+2N_l\right)$ & $N_{\rm c} \left(3N_l^3+2N_l^2M_l+2N_lM_l^2+M_l^3+N_l^2+N_l\right)$ & $O(N_{\rm c} N_l^3)$\\
SS-DKF & $N_{\rm c} \left(2N_l^2+2N_lM_l+2N_l \right) $ & $N_{\rm c} \left(3N_l^2+2N_lM_l+N_l\right)$ & $O(N_{\rm c} N_l^2)$ \\
\bottomrule
\end{tabular}
\label{Tab:Complexity-Comparison}
\end{table*}

\subsection{Steady-State Distributed Kalman Filtering}\label{Sec:SS-DKF-algorithm}

{\color{blue}
An opportunity can be identified that, if some mild reduction is introduced for the considered thermal model, we can obtain another DKF approach with much higher computational efficiency. To be specific, consider the heat generation equation~\eqref{heat-generation}. Many studies in the literature suggest that its second term often has a negligible magnitude in comparison with the first term and thus can be ignored~\cite{lin2014lumped,kim2014estimation}. With this simplification,~\eqref{heat-generation} can be reduced as $q={I\left(U_{\rm ocv}-U_{\rm t}\right)}/{V_{\rm co}}$. It is then found that $\bm F$ becomes time-invariant in this case, which will allow us to develop a more computationally efficient DKF for temperature field reconstruction. The development is as follows.
}

Before proceeding further, we make an assumption at first.
\begin{assumption}\label{assumption1}
The pair $\left(\bm{F}_{ll},\bm{H}_l\right)$ is detectable and the pair $\left(\bm{F}_{ll}, \bm{Q}_l^{1\over 2} \right)$ stabilizable for $l\in \{1,\cdots, N_{\rm c}\}$.
\end{assumption}
If Assumption~\ref{assumption1} holds, the DKF algorithm in~\eqref{DKF-prediction1}-\eqref{DKF-filter3} will achieve steady state. Specifically, $\bm{P}_{k|k-1}^{(l)}$ will converge to a unique stabilizing solution, $\bar{\bm{P}}^{(l)}$, of the discrete algebraic Riccati equation
\begin{align*}
\bm{X}= &\bm{F}_{ll}\bm{X}\bm{F}_{ll}^\T - \bm{F}_{ll} \bm{X}\bm{H}_l^\T \left( \bm{H}_l \bm{X} \bm{H}_l^\T+\bm{R}_l \right)^{-1} \bm{H}_l \bm{X} \bm{F}_{ll}^\T +\bm{Q}_l ,
\end{align*}
where $\bm{X}$ is an unknown symmetric positive-definite matrix with compatible dimensions. The gain matrix $\bm{K}_k^{(l)}$ in~\eqref{DKF-filter2} will consequently reach a fixed point
\begin{align}\label{SS-DKF-gain}
\bar{\bm{K}}^{(l)} = \bar{\bm{P}}^{(l)} \bm{H}^{\T}_{l} \left(\bm{H}_l \bar{\bm{P}}^{(l)} \bm{H}^{\T}_l + \bm{R}_l \right)^{-1},
\end{align}
which can ensure $\bm{F}_{ll}(\bm{I}_l-\bar{\bm{K}}^{(l)}\bm{H}_l)$ to be stable~\cite{anderson2015optimal}. With fixed $\bar{\bm P}^{(l)}$ and $\bar{\bm{K}}^{(l)}$, the state prediction and update can be accomplished more efficiently:
\begin{align}\label{SS-DKF}
\hat{\bm{x}}_{k|k-1}^{(l)} &= \bm{F}_{ll} \hat{\bm{x}}_{k-1|k-1}^{(l)} + \sum_{i\in \calN_l} \bm{F}_{li}\hat{\bm{x}}_{k-1|k-1}^{(i)}+\bm{G}_{l} \bm{u}_{k-1} ,\\ \label{SS-DKF2}
\hat{\bm{x}}_{k|k}^{(l)} &= \hat{\bm{x}}_{k|k-1}^{(l)} + \bar{\bm{K}}^{(l)} \left( \bm{y}_k^{(l)} - \bm{H}_l \hat{\bm{x}}_{k|k-1}^{(l)} \right),
\end{align}
which together form the SS-DKF algorithm. It is seen that the SS-DKF does not maintain the estimation error covariance and that its gain matrix can be computed offline prior to the estimation run. Although this comes at certain sacrifice of estimation accuracy, it still maintains stability under some assumptions and presents much appeal from a computational perspective. Before moving on to the computational complexity analysis in Section~\ref{Sec:Complexity-Analysis}, the stability of the SS-DKF algorithm is examined in the remainder of this section. 

We define the real state error $\bm{e}_k = \hat{\bm{x}}_{k|k}-\bm{x}_k $. For the real state error $\bm{e}_k^{(l)}$ at subsystem $l$, combining \eqref{SS-DKF} and \eqref{SS-DKF2}, we have
\begin{align}\label{SS-DKF-error}\nonumber
\bm{e}_k^{(l)} = & (\bm{I}-\bar{\bm{K}}^{(l)}\bm{H}_l)\bm{F}_{ll}\bm{e}_{k-1}^{(l)}+\sum_{i\in \mathcal{N}_l}(\bm{I}- \bar{\bm{K}}^{(l)}\bm{H}_l)\bm{F}_{li}\bm{e}_{k-1}^{(i)}\\&
-(\bm{I}-\bar{\bm{K}}^{(l)}\bm{H}_l)\bm{w}_{k-1}^{(l)}+\bar{\bm{K}}^{(l)}\bm{v}_k^{(l)},
\end{align}
where $\bm{I} \in \mathbb{R}^{N_l \times N_l}$. Aggregating $\bm{e}_k^{(l)}$ for $l\in \{ 1,\cdots, N_{\rm c}\}$ together yields
\begin{align}\label{SS-DKF-error2}
\bm{e}_k = 
(\bm{I}-\bar{\bm{K}}\bm{H})\bm{F} \bm{e}_{k-1}-(\bm{I}-\bar{\bm{K}}\bm{H})\bm{w}_{k-1}+\bar{\bm{K}}\bm{v}_k,
\end{align}
where $\bar{\bm{K}} = \mathrm
{blkdiag}\left\{\bar{\bm{K}}^{(1)},\cdots,\bar{\bm{K}}^{(N_{\rm c})}\right\}$ and identity matrix $\bm{I} \in \mathbb{R}^{N \times N}$. 

The stability of the SS-DKF algorithm is summarized as follows.

\begin{theorem}\label{convergence-theorem}
Let $\bm A = (\bm{I}-\bar{\bm{K}}\bm{H})\bm{F}$ and $\bm B = (\bm{I}-\bar{\bm{K}}\bm{H})\bm{Q}(\bm{I}-\bar{\bm{K}}\bm{H})^{\top}+\bar{\bm{K}}\bm{R}\bar{\bm{K}}^{\top}$.
If $\bm A$ is stable, the true error covariance ${\bm \Sigma}_k = \mathbf{E}\left[\bm e_k \bm e_k^{\T}\right]$ will converge to the unique  solution of the discrete-time Lyapunov equation
\begin{align}\label{lyapunov}
\bm \Sigma = \bm A \bm \Sigma \bm A^\T +\bm B.
\end{align}
\end{theorem}
\begin{proof}
It is seen that the propagation of $\bm \Sigma_k$ is governed by
\begin{align}\label{DKF-covariance}
\bm \Sigma_{k} = 
\bm A\bm \Sigma_{k-1}\bm{A}^\T  +\bm B.
\end{align}
Since $\bm A$ is stable, $\lim_{k\to\infty}\bm \Sigma_{k}=\bm \Sigma$ (see Chapter 3.3 in~\cite{varaiya1986stochastic}). 
\end{proof}

From above, Assumption~\ref{assumption1} lays the foundation for the derivation of the SS-DKF, and Theorem~\ref{convergence-theorem} indicates that a stable $(\bm{I}-\bar{\bm{K}}\bm{H})\bm{F}$ can guarantee the stability of the SS-DKF algorithm. However, a question then arises: will Assumption~\ref{assumption1} and the stability of $(\bm{I}-\bar{\bm{K}}\bm{H})\bm{F}$ hold for the LiB pack model in~\eqref{new-system}? 

An examination is given as follows. First, consider Assumption~\ref{assumption1}. Note that the thermal physics implies that the model established in Section~\ref{Sec:Modeling} is stable if the LiB pack operates normally and the numerical stability criterion is satisfied in discretization. 
Hence, $\bm F$ will be stable. Next, we partition $\bm{F}$ into the following form:
\begin{align}\label{F-decomposition}
\bm{F} = \bm{F}_{\rm d}+\bm{F}_{\rm od},
\end{align}
where the subscripts $\rm d$ and $\rm od$, respectively, denote diagonal blocks and off-diagonal blocks and $\bm{F}_{\rm d}=\textrm{blkdiag}\{\bm{F}_{11},\cdots,\bm{F}_{N_{\rm c},N_{\rm c}}\}$. According to Corollary 5.6.14 in~\cite{horn1987matrix}, one will have
\begin{align}\label{radius-of-Fod}
\rho(\bm{F}_{\rm od})=\textrm{lim}_{k\to \infty} \| \bm{F}_{\rm od}^k \|^{1/k}.
\end{align}
It is interesting to note that $\bm{F}_{\rm od}^2 = \bf 0$ in this application due to the serial structure of the LiB pack, which will be further shown in Section~\ref{Sec:Simulation} . Therefore, $\rho(\bm{F}_{\rm od})=0$ in this case. 
Invoking Lemma 5.6.10 in~\cite{horn1987matrix}, there exists a matrix norm $\|{\cdot}\|_*$ for any given $\epsilon>0$ such that 
\begin{align}\label{lyapunov2}
\rho( {\bm{F}_{\rm d} }) \le \|{ {\bm{F}_{\rm d} }}\|_* \le \rho({\bm{F}_{\rm d} }) +\epsilon.
\end{align}
Since it satisfies the triangle inequality, then
\begin{align}\nonumber 
\|{\bm{F}_{\rm d} }\|_* & \le \|{ \bm{F}}\|_* + \|{\bm{F}_{\rm od} }\|_* \\ \label{norm-inequality4} & \le
(\rho(\bm{F})+\epsilon_1) +(\rho(\bm{F}_{\rm od})+\epsilon_2),
\end{align}
where $\epsilon_1>0$ and $\epsilon_2>0$. Since $\rho(\bm{F}_{\rm od})=0$, \eqref{norm-inequality4} can be rewritten as
\begin{equation}\label{norm-inequality3}
\begin{aligned}
\|{\bm{F}_{\rm d} }\|_* \le
\rho(\bm{F})+\epsilon_1 +\epsilon_2.
\end{aligned}
\end{equation}
Thus, one can always find $\epsilon_1$ and $\epsilon_2$ to prove that $\bm{F}_{\rm d}$ is stable, thus validating Assumption~\ref{assumption1}.

Now, consider the stability of $(\bm{I}-\bar{\bm{K}}\bm{H})\bm{F}$. 
For notational simplicity, we denote $\bm{I}-\bar{\bm{K}}\bm{H}$ as $\tilde{\bm{I}}$. Then the objective is to show that $\tilde{\bm{I}}\bm{F}$ is stable. Recalling the matrix norm $\|{\cdot}\|_*$ in~\eqref{lyapunov2}, it is also submultiplicative and implies
\begin{align}\nonumber
\|{\tilde{\bm{I}}\bm{F} }\|_* & \le \|{ \tilde{\bm{I}}\bm{F}_{\rm d} }\|_* + \|{\tilde{\bm{I}}\bm{F}_{\rm od} }\|_* \\ \nonumber &
\le \|{ \tilde{\bm{I}}\bm{F}_{\rm d} }\|_* + \|{ \tilde{\bm{I}}}\|_* \|{\bm{F}_{\rm od} }\|_* \\ \nonumber
& \le
(\rho(\tilde{\bm{I}}\bm{F}_{\rm d})+\epsilon_3) +(\rho(\tilde{\bm{I}})+\epsilon_4)(\rho(\bm{F}_{\rm od})+\epsilon_5) \\ \label{norm-inequality2}& \le 
\rho(\tilde{\bm{I}}\bm{F}_{\rm d})+\epsilon_3+\rho(\tilde{\bm{I}})\epsilon_5+\epsilon_4 \epsilon_5,
\end{align}
where $\epsilon_3>0$, $\epsilon_4>0$, and $\epsilon_5>0$. 
Then, because $\bm{F}_{ll}(\bm{I}_l-\bar{\bm{K}}^{(l)}\bm{H}_l)$ is stable, the matrix $\bm{F}_d\tilde{\bm{I}}$ is stable. Following that $\lim_{k\to \infty}(\bm{F}_{\rm d}\tilde{\bm{I}})^k=\bf{0}$, $(\tilde{\bm{I}}\bm{F}_{\rm d})^k$ can be constructed as $\tilde{\bm{I}}(\bm{F}_{\rm d}\tilde{\bm{I}})^{k-1}\bm{F}_{\rm d}$, such that $\lim_{i\to \infty}(\tilde{\bm{I}}\bm{F}_{\rm d})^k=\mathbf{0}$. Therefore, it is always possible to find $\epsilon_3$, $\epsilon_4$ and $\epsilon_5$ to make the right-hand side of~\eqref{norm-inequality2} smaller than 1. Subsequently, $\rho(\tilde{\bm{I}}\bm{F})<1$ and $(\bm{I}-\bar{\bm{K}}\bm{H})\bm{F}$ is stable. 

{\color{blue}
\begin{remark}
{\rm (Extension to thermal runaway detection).} In above, the DKF and SS-DKF are developed to reconstruct the temperature field of a battery pack. They can be used as a tool to monitor the spatially distributed thermal behavior critical for a battery pack's safety. An extension of them to detect thermal runaway can be hopefully made. An idea is to consider the thermal runaway as an unknown disturbance that abruptly appears and applies to the model in~\eqref{new-system}. Then, the thermal runaway detection can be formulated as the problem of disturbance detection. KF-based approaches have been studied extensively for disturbance detection in the literature, e.g.,~\cite{gustafsson2000adaptive}, and can be potentially exploited here. Combining this idea and the design in this paper, we can promisingly build distributed KF-based approaches for thermal runaway detection. This will be an important part of our future work.
\end{remark}
}

\subsection{Computational Complexity Analysis}\label{Sec:Complexity-Analysis}

As aforementioned, the objective of distributing the CKF across the cells is to improve the computational efficiency toward enabling real-time reconstruction of a LiB pack's temperature field. In this subsection, the CKF, DKF and SS-DKF algorithms are analyzed and compared in terms of the computational complexity. The analysis is based on the number of arithmetic operations needed by each algorithm. Before proceeding further, let us consider the basic matrix operations. 
For two $n\times n$ matrices, their addition involves $n^2$ elementary additions, and their multiplication involves $n^3$ elementary multiplications and $(n-1)n^2$ elementary additions. The inverse of an $n\times n$ matrix requires $n^3$ elementary multiplications and $n^3$ elementary additions. These are the basic algebraic arithmetics involved in the considered algorithms. The complexity of each algorithm can be assessed by summing up all the arithmetic operations required at each time instant, with the results shown in Table~\ref{Tab:Complexity-Comparison}.

It is demonstrated in Table~\ref{Tab:Complexity-Comparison} that the CKF has the heaviest computation at $O \left(N^3 \right)$, which increases cubically with the size of the state space of the entire pack. This also implies that, when a pack has more cells, the computation would rise cubically with the cell number. The computational complexity at such a level can be unaffordable by a real-world onboard computing platforms. By comparison, the DKF is much more efficient than the CKF. Given that $N= N_{\rm c} N_l$ with $N_{\rm c} \ll N_l$, the DKF's arithmetic operations at $O(N_{\rm c} N_l^3)$ are only about one $N_{\rm c}^2$-th of the CKF's. In addition, with the computation increasing only linearly with the cell number, the DKF well lends itself to parallel processing, where the estimation for each cell is performed on a separate micro-processor at a complexity $O(N_l^3)$. In this scenario, an increase in the cell number will not add cost to the existing micro-processors. The SS-DKF unsurprisingly is the most computationally competitive. Its complexity at $O\left( N_{\rm c} N_l^3 \right)$ is even one order less than that of the DKF. Just like the DKF, it is also well suited for execution based on parallel processing.

\section{Numerical Simulation}\label{Sec:Simulation}

To reconstruct a LiB pack's temperature field, the previous sections presented a spatially resolved thermal model and distributed the KF for computationally fast estimation. {\color{blue}In this section, numerical simulation with a practical LiB pack is offered to illustrate the effectiveness of the proposed work. The simulation is performed using MATLAB.}

\begin{table}\centering
\begin{threeparttable}
\caption{LiB cell parameters used in the simulation~\cite{chen2005thermal}.}
\begin{tabular}{ccccc}
\toprule%
Layer/ & Thickness & Density & Heat capacity & Conductivity \\
Material & [cm] & [kg/m$^3$] & [J/(kg$\cdot$K)] & [W/(m$\cdot$K)] \\
\midrule
Anode$^{\star}$ & 0.0116 & 1347.33 & 1437.4 & 1.04 \\
Cathode$^{\star}$ & 0.014 & 2328.5 & 1269.21 & 1.58 \\
Separator$^{\star}$ & 0.0035 & 1008.98 & 1978.16 & 0.3344 \\
Cu foil$^{\star}$ & 0.0014 & 8933 & 385 & 398 \\
Al foil$^{\star}$ & 0.002 & 2702 & 903 & 238 \\
Metal case$^{\ast}$ & 0.07 & 2770 & 875 & 170 \\ 
Contact layer$^{\ast}$ & 0.05 & 1129.95 & 2055.1 & 0.60 \\
\bottomrule
\end{tabular}
\begin{tablenotes}
\small
\item Note: $\star$~component of cell units, and $\ast$~component of cell case.
\end{tablenotes}
\label{Tab:Cell-Parameters}
\end{threeparttable}
\end{table}

\subsection{Simulation Setting}\label{Sec:simulation-setting}
Consider a LiB pack that consists of three large-format high-capacity prismatic LiB cells connected and stacked in series. Here, the cells are the same ones as in~\cite{chen2005thermal}. Each cell has a capacity of 185.3 Ah and is 19.32 cm long, 10.24 cm wide, and 10.24 cm high. {\color{blue}As mentioned in Section~\ref{Sec:Modeling}, the cell has two portions: the core region and the metallic case. The core region is 19.08 cm long, 10 cm wide, and 10 cm high, housing three hundred smaller cell units connected in parallel. The structure of a cell unit is schematically shown in Figure~\ref{Fig:Unit-Schematic}.} The key parameters of the LiB cell are summarized in Table~\ref{Tab:Cell-Parameters}. In the simulation, it is assumed that the LiB pack operates in an environment with temperature maintained at 300 K. The convective heat transfer coefficient, emissivity on pack surface, and entropic heat transfer coefficient $\textrm{d} U_{\rm ocv}/ \textrm{d} T$ are set to be 30 W/(m$^2\cdot$K), 0.25, and 0.00022 VK$^{-1}$~\cite{chen2005thermal}, respectively. {\color{blue}The battery pack is discharged using a time-varying current profile, which is shown in Figure~\ref{Fig:Crate} and derived from the Urban Dynamometer Driving Schedule (UDDS)~\cite{UDDS}.} {\color{blue}Sensors are mounted on the battery pack as shown in Figure~\ref{Fig:grid-of-pack}. That is, five sensors are placed on each cell-air interface of a cell. Such a placement is straightforward and easy to implement. Associated with this, an intriguing question is how to optimally deploy the sensors toward achieving satisfactory estimation performance with a minimum number of sensors. While some results are reported in the literature, e.g.,~\cite{lin2014temperature,fang2014optimal,zhang2017sensor}, further research is still required to fully address this question.
}

\begin{figure}[t]
\centering
\includegraphics[trim = {20mm 65mm 30mm 30mm}, clip, width=0.43\textwidth]{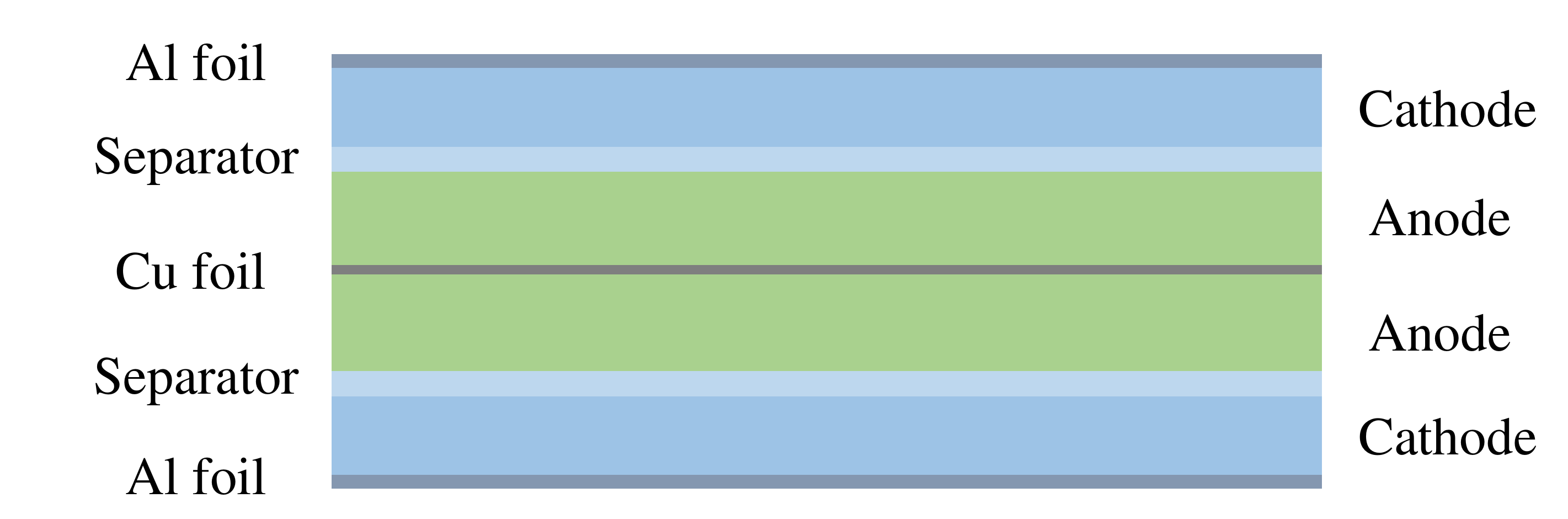}
\centering
\caption{Schematic diagram of a basic unit in a LiB cell.}
\label{Fig:Unit-Schematic}
\end{figure}


\begin{figure}[t]
\centering
\includegraphics[trim = {10mm 0mm 5mm 0mm}, clip, width=0.43\textwidth]{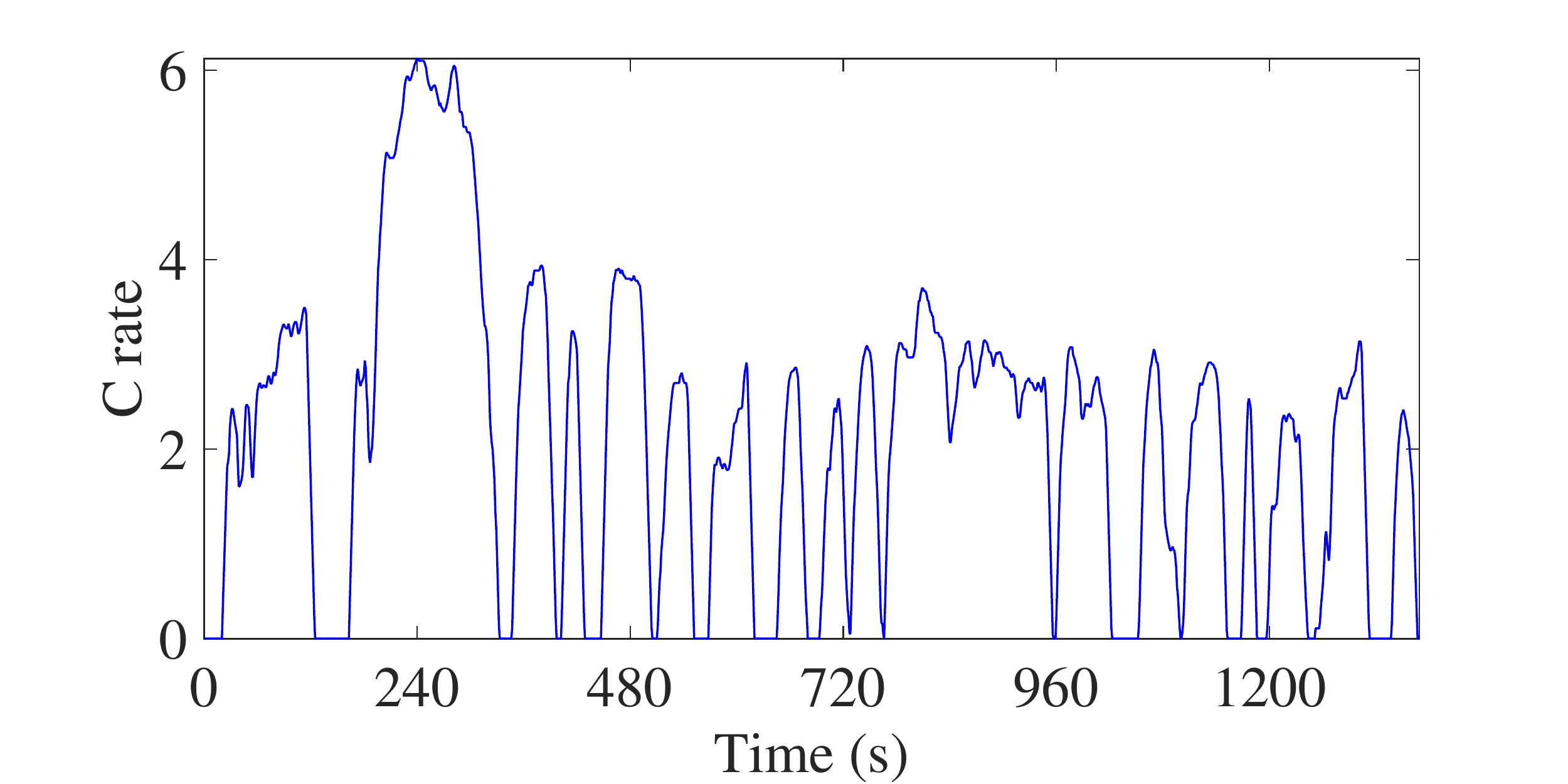}
\centering
\caption{\color{blue}Discharging current profile based on UDDS.}
\label{Fig:Crate}
\end{figure}

The thermal dynamics of the considered pack can be characterized by the PDE-based model in Section~\ref{Sec:Modeling}. Each cell is gridded in space with $m = 9$, $n = 5$ and $p = 5$ and in time with $\Delta t = 1~{\mathrm {s}}$. {\color{blue}In general, one can increase $m$, $n$ and $p$ and reduce $\Delta t$ to increase the accuracy of simulation. This, however, will come at the sacrifice of computational efficiency. Another risk lies in numerical instability, which can be caused if the selected $m$, $n$, $p$ and $\Delta t$ fail to satisfy certain conditions~\cite{bergman2006fundamentals}. To find a satisfactory set, one can consider a few candidates. S/he can first check the numerical stability for each set using the conditions in~\cite{bergman2006fundamentals}. Then, run the simulation for the sets that pass the check, and choose the set that leads to acceptable accuracy with minimum computational cost. This process understandably may require repeated trial effort.}

By discretization, the PDEs are converted to a state-space model as shown in Section~\ref{Sec:Model-Conversion}. On each cell-air interface of the cell, five sensors are deployed as shown in Figure~\ref{Fig:grid-of-pack}. It follows that the model matrices $\bm{F}$ and $\bm{H}$ for this pack take the following form:
\begin{align*}
\bm{F} = \left[ \begin{array}{ccc}
\bm{F}_{11} & \bm{F}_{12} & \\
\bm{F}_{21} & \bm{F}_{22} & \bm{F}_{23} \\
& \bm{F}_{32} & \bm{F}_{33}
\end{array} \right] \in \mathbb{R}^{675\times 675},\\
\bm{H} = \left[ \begin{array}{ccc}
\bm{H}_1 & & \\
& \bm{H}_2 & \\
& & \bm{H}_3 
\end{array} \right] \in \mathbb{R}^{70\times 675}.
\end{align*}
In the simulation, the pack's initial temperature is 300 K, the same with the ambient temperature. Yet, for the purpose of illustrating the estimation, the initial guess is 295 K in the simulation. The noise covariance matrices $\bm{Q}$ and $\bm{R}$ are chosen as
$$
\bm{Q} = 0.05^2 \cdot  \bm{I}^{675\times 675}, \ \bm{R} = 0.3^2 \cdot \bm{I}^{70\times 70},
$$
where $\bm{I}$ is the identity matrix. 

\begin{figure*} [t]
\centering 
\subfigure
{
\includegraphics[trim = {30mm 10mm 10mm 10mm}, clip, width=0.23\textwidth]{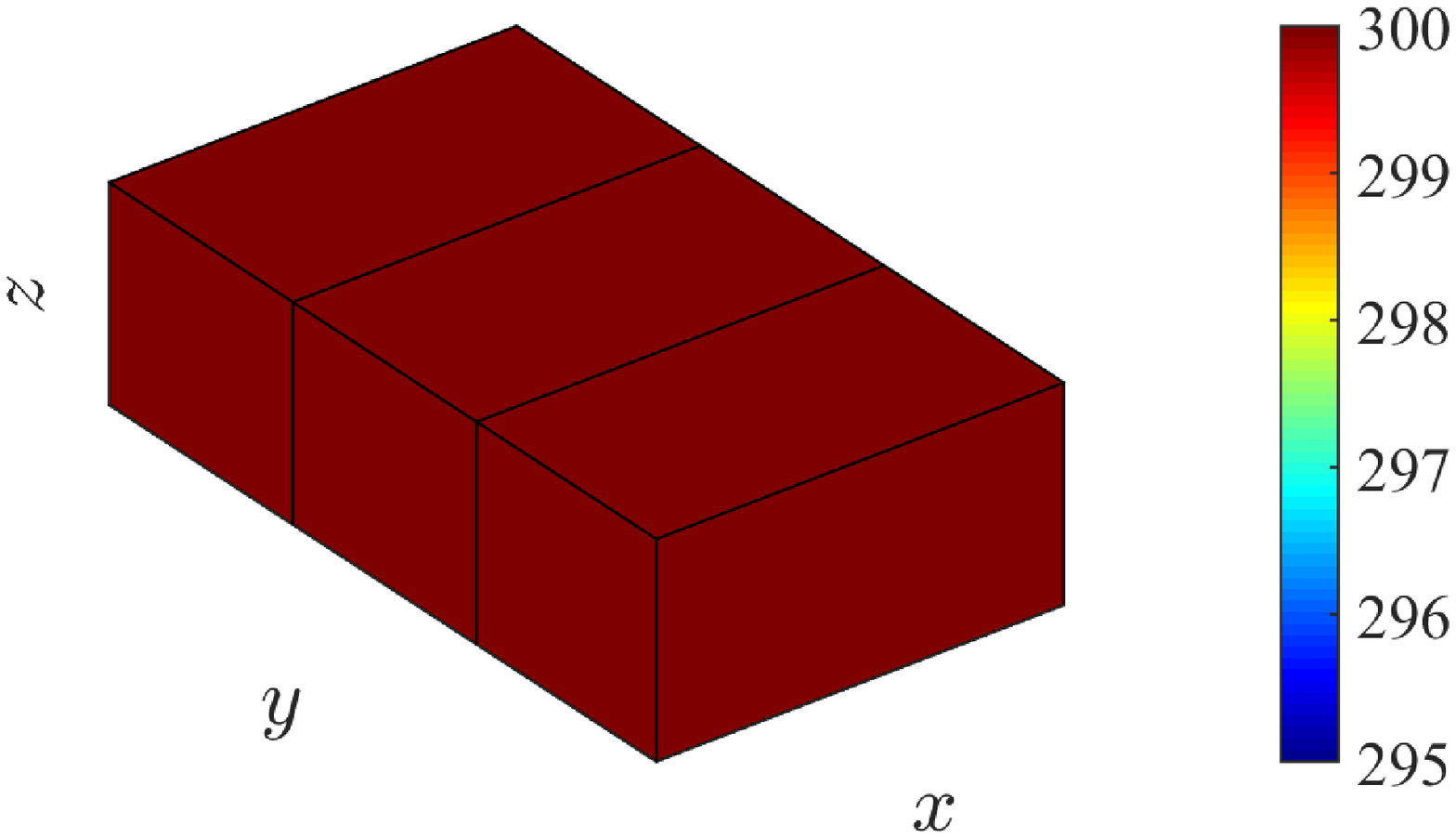}} 
\hspace{0in} 
\subfigure
{
\includegraphics[trim = {30mm 10mm 10mm 10mm}, clip, width=0.23\textwidth]{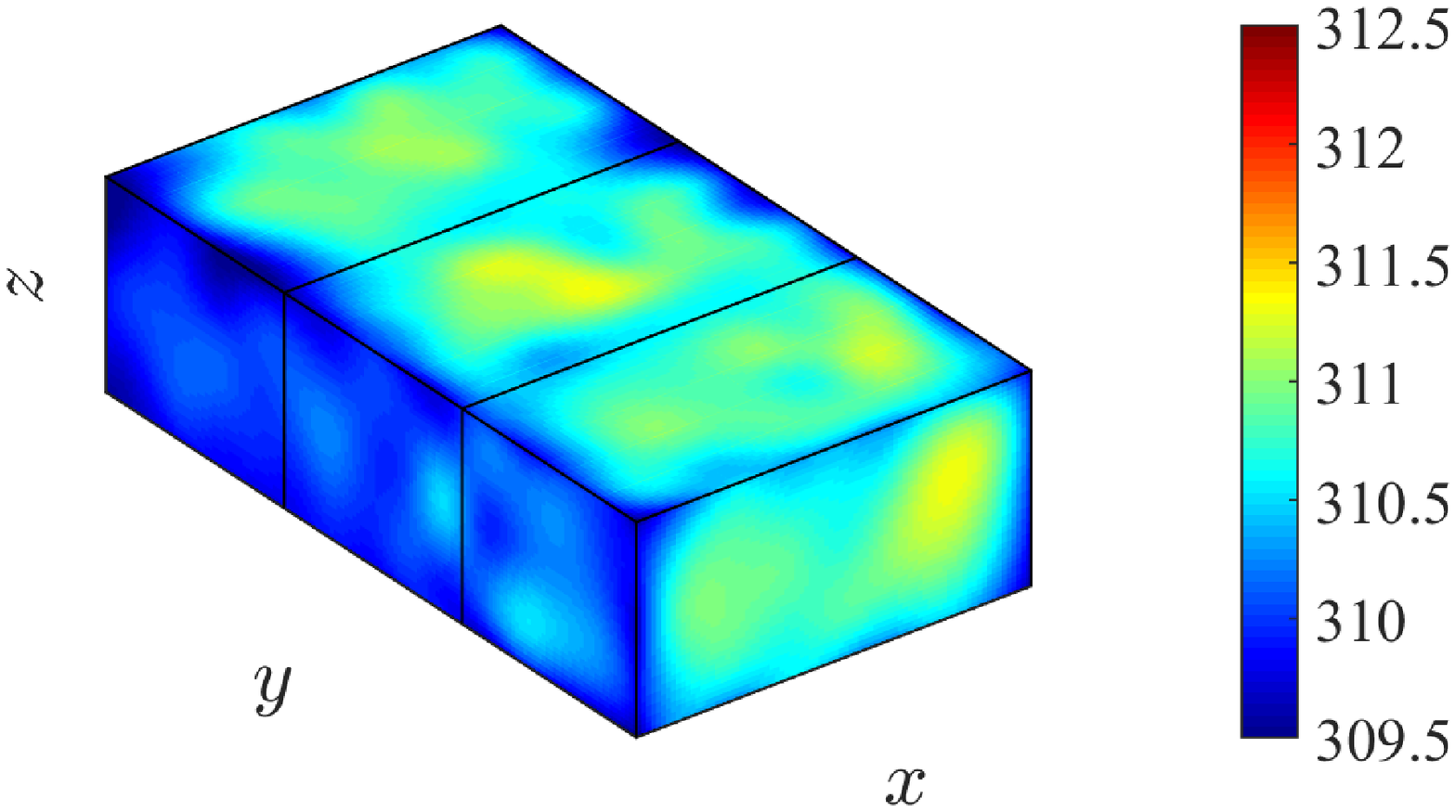}} 
\hspace{0in} 
\subfigure
{ 
\includegraphics[trim = {30mm 10mm 10mm 10mm}, clip, width=0.23\textwidth]{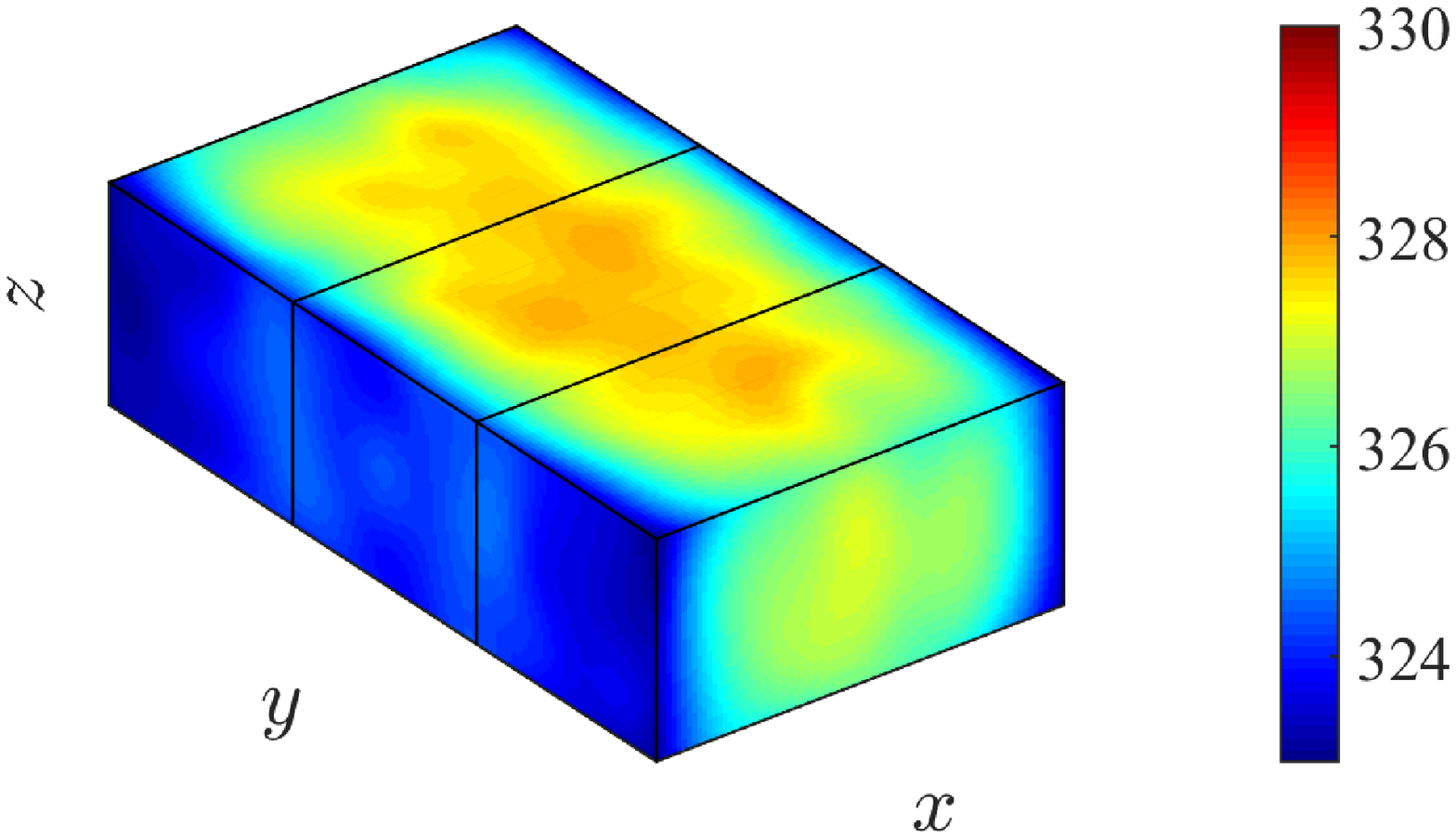}} 
\hspace{0in} 
\subfigure
{
\includegraphics[trim = {30mm 10mm 10mm 10mm}, clip, width=0.23\textwidth]{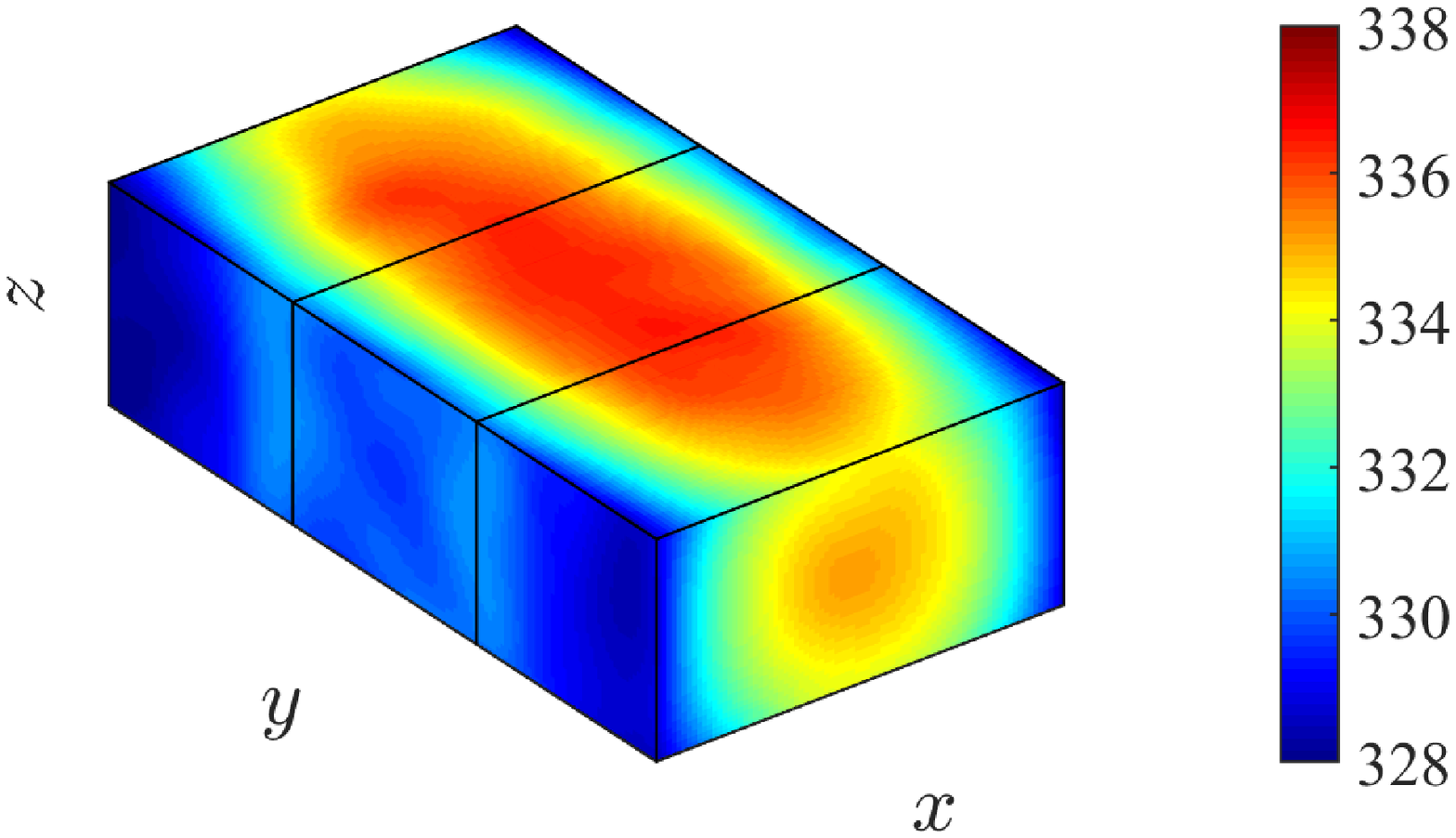}}\\
\subfigure
{
\includegraphics[trim = {30mm 10mm 10mm 10mm}, clip, width=0.23\textwidth]{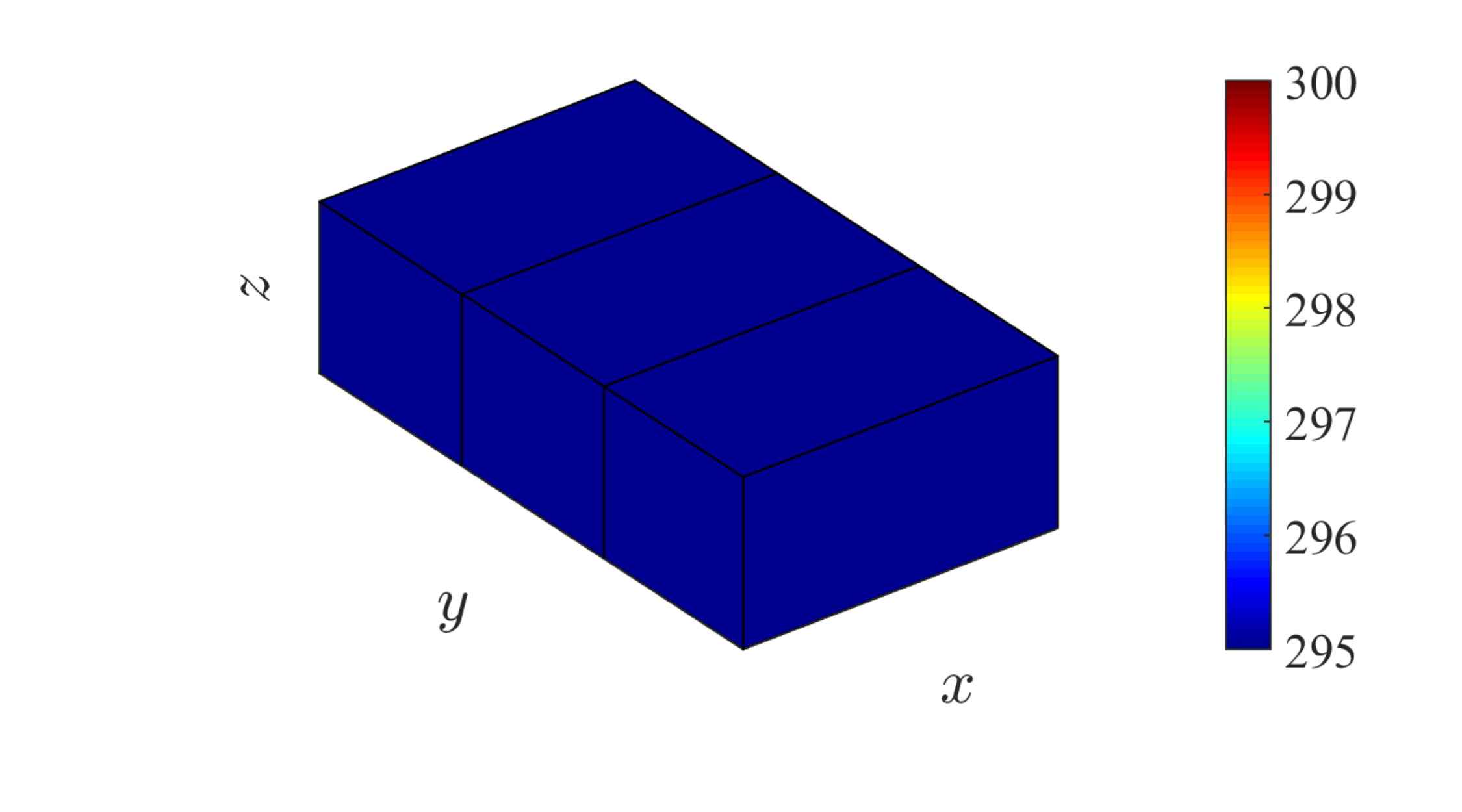}} 
\hspace{0in} 
\subfigure
{
\includegraphics[trim = {30mm 10mm 10mm 10mm}, clip, width=0.23\textwidth]{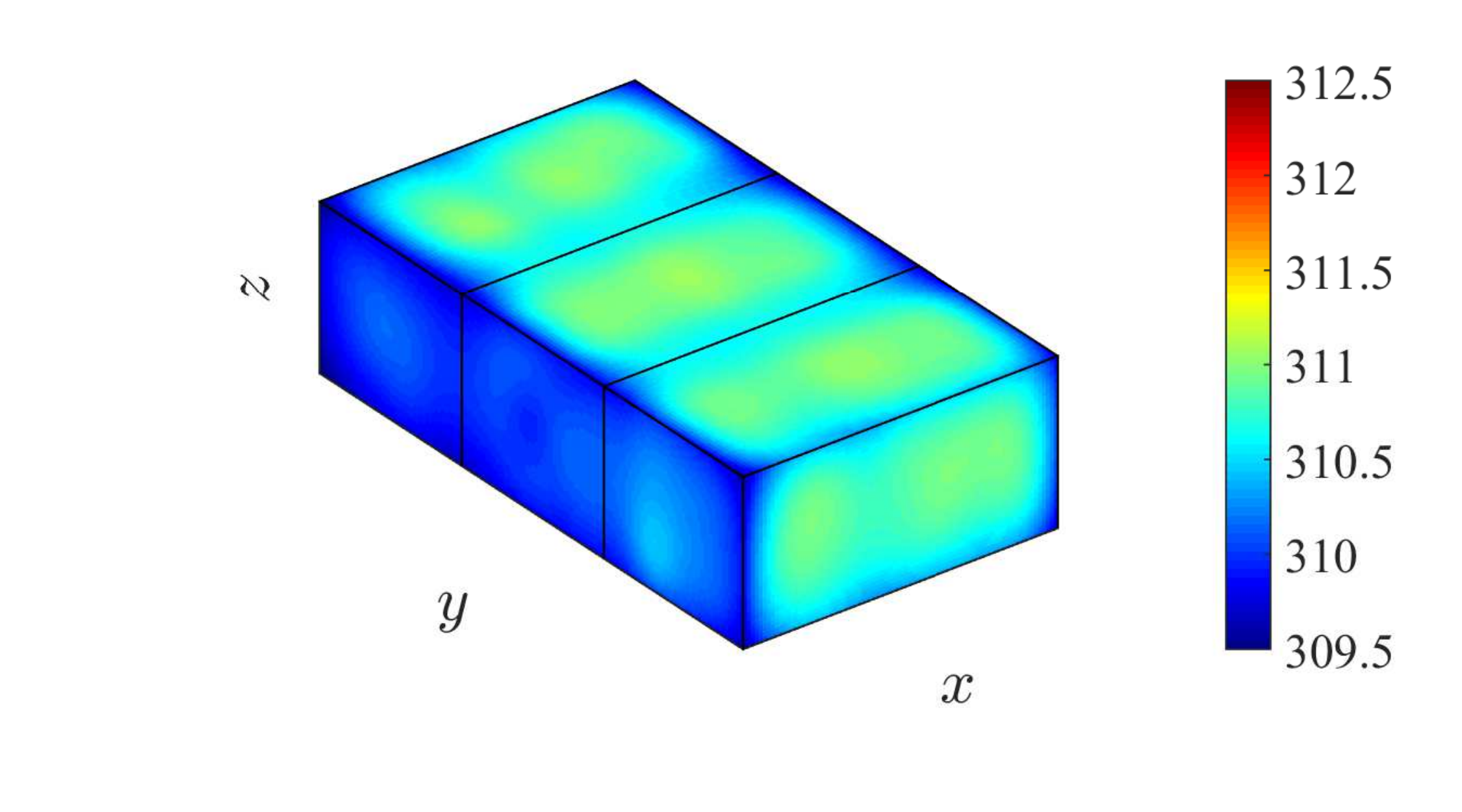}} 
\hspace{0in} 
\subfigure
{ 
\includegraphics[trim = {30mm 10mm 10mm 10mm}, clip, width=0.23\textwidth]{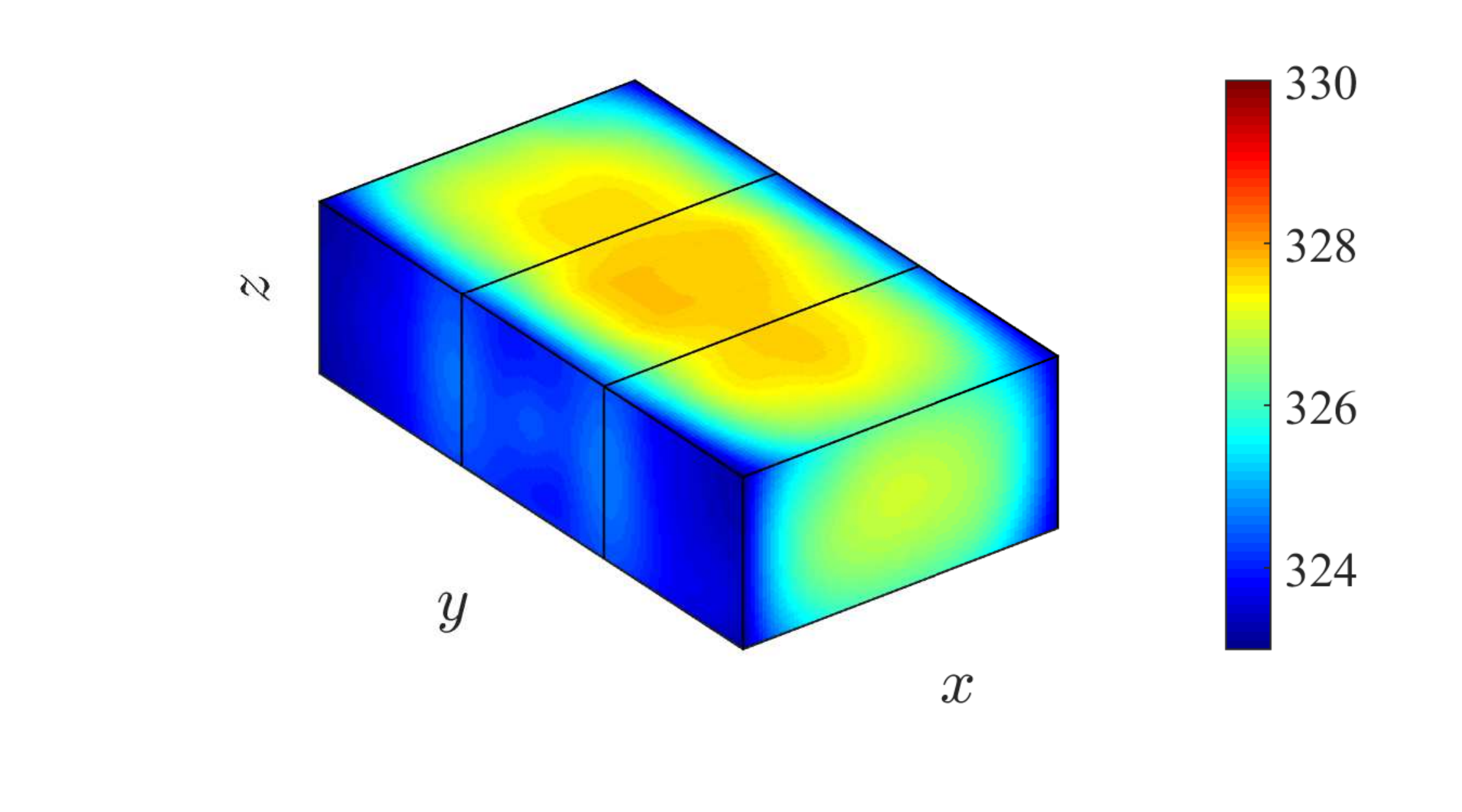}} 
\hspace{0in} 
\subfigure
{
\includegraphics[trim = {30mm 10mm 10mm 10mm}, clip, width=0.23\textwidth]{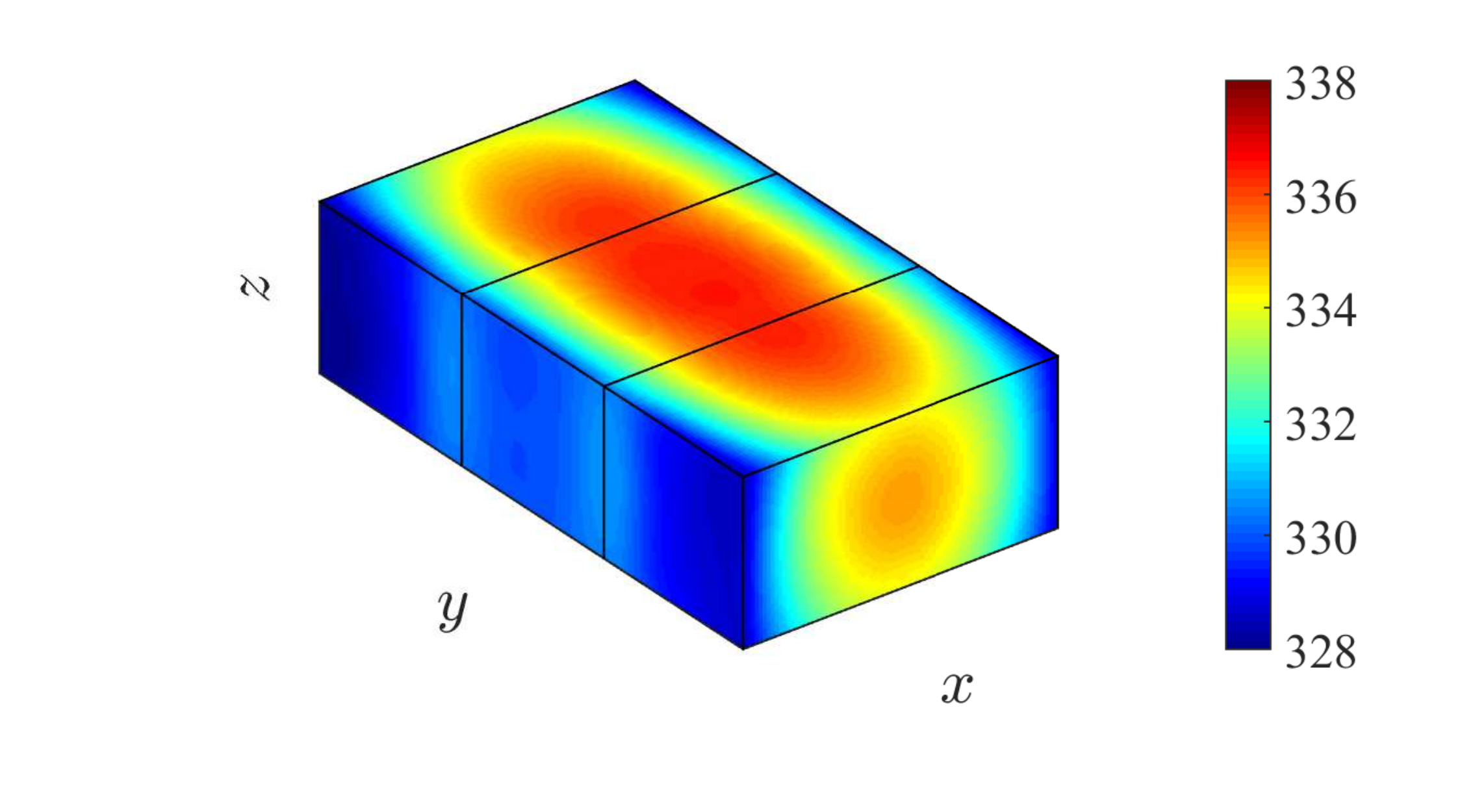}}\\
\subfigure[$t$ = 0 s]
{
\addtocounter{subfigure}{-8} 
\includegraphics[trim = {30mm 10mm 10mm 10mm}, clip, width=0.23\textwidth]{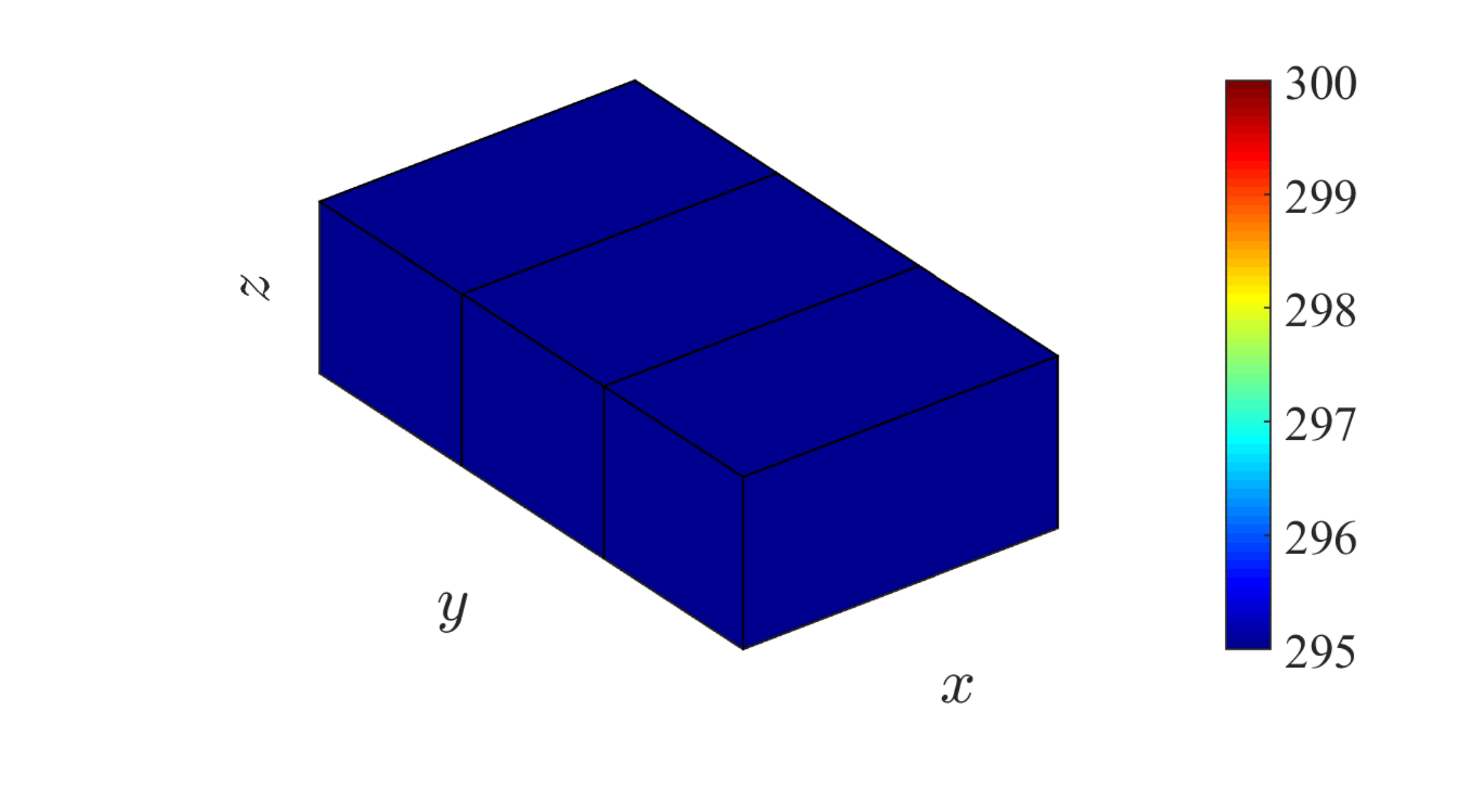}} 
\hspace{0in} 
\subfigure[$t$ = 240 s]
{
\includegraphics[trim = {30mm 10mm 10mm 10mm}, clip, width=0.23\textwidth]{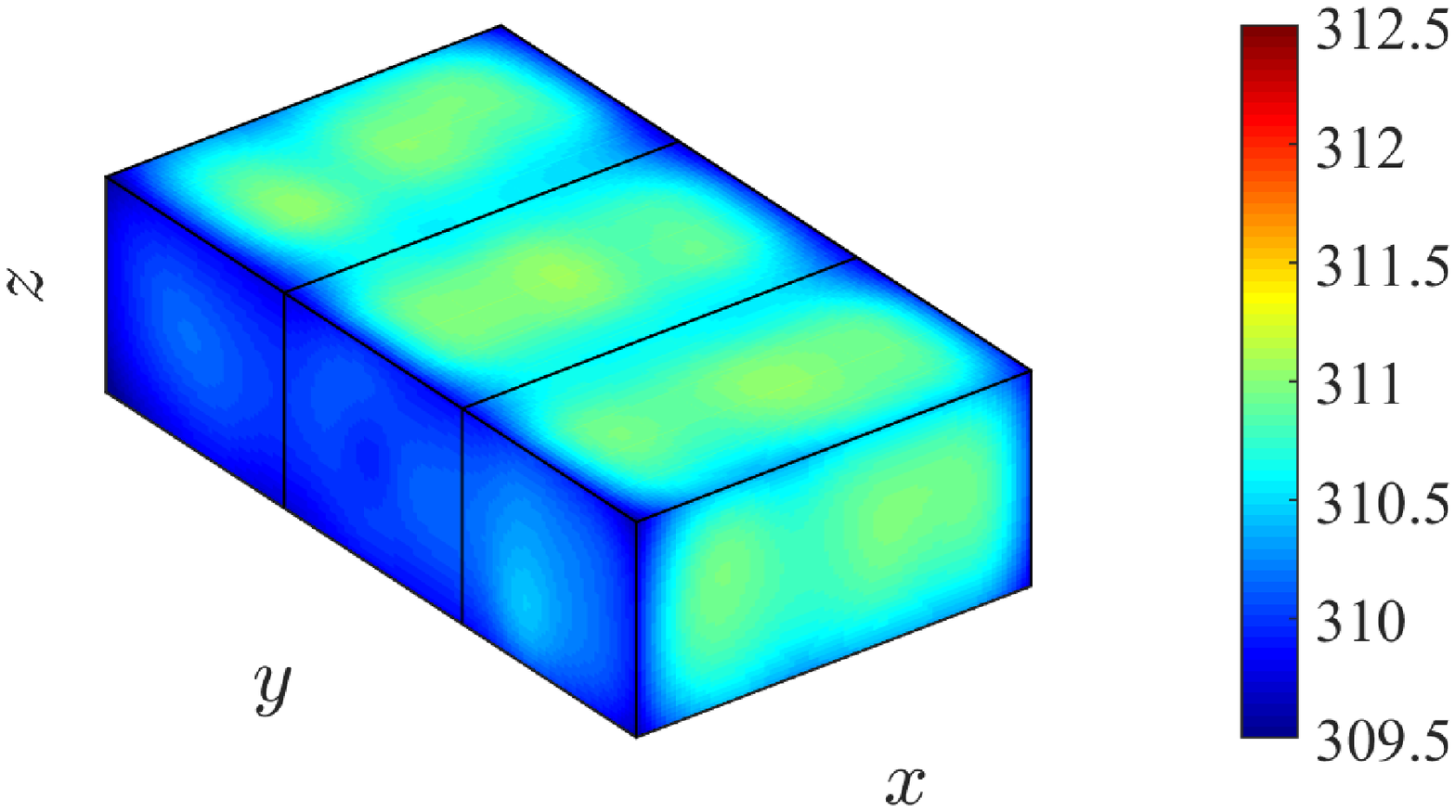}} 
\hspace{0in} 
\subfigure[$t$ = 480 s]
{ 
\includegraphics[trim = {30mm 10mm 10mm 10mm}, clip, width=0.23\textwidth]{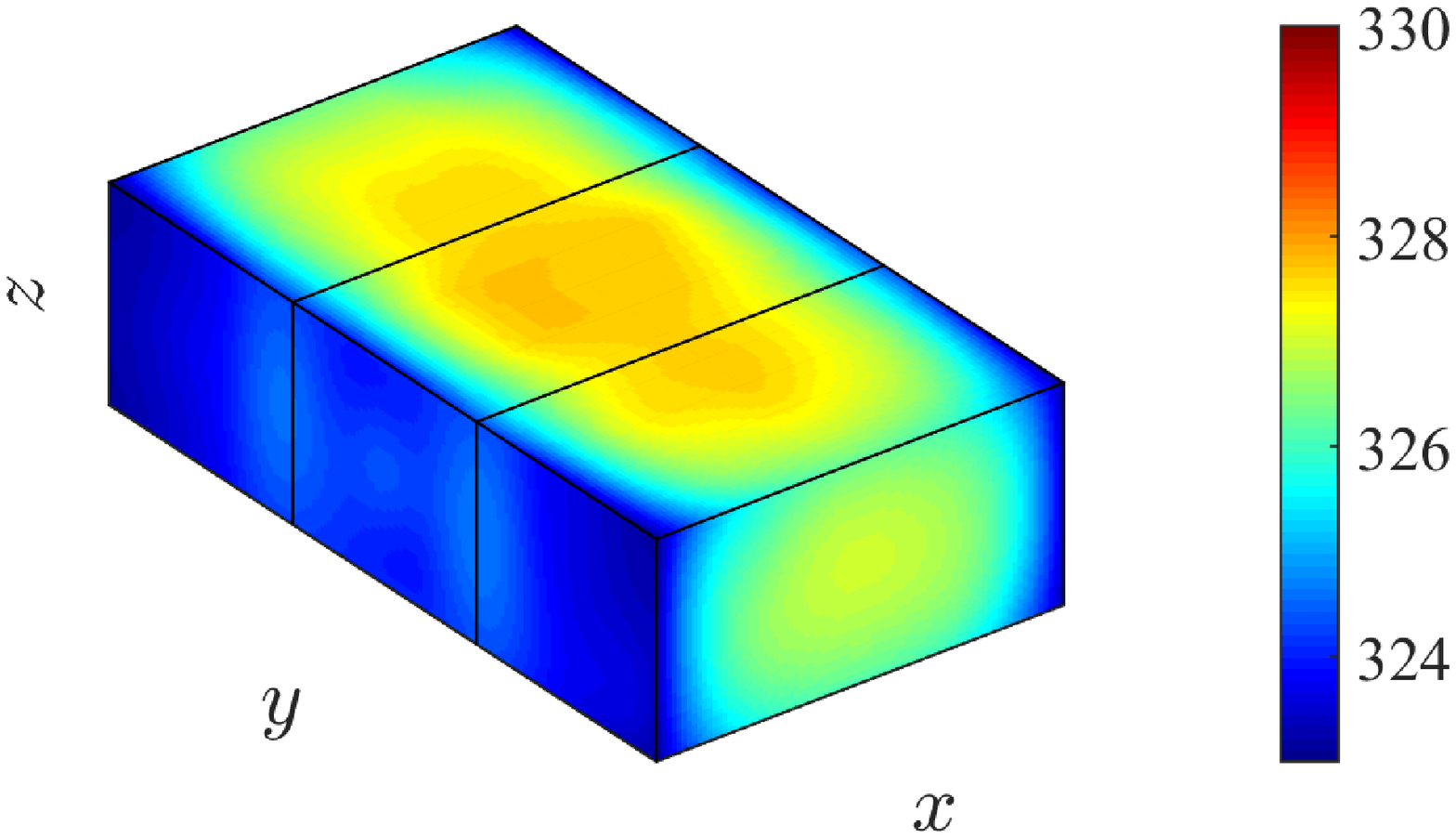}} 
\hspace{0in} 
\subfigure[$t$ = 1200 s]
{
\includegraphics[trim = {30mm 10mm 10mm 10mm}, clip, width=0.23\textwidth]{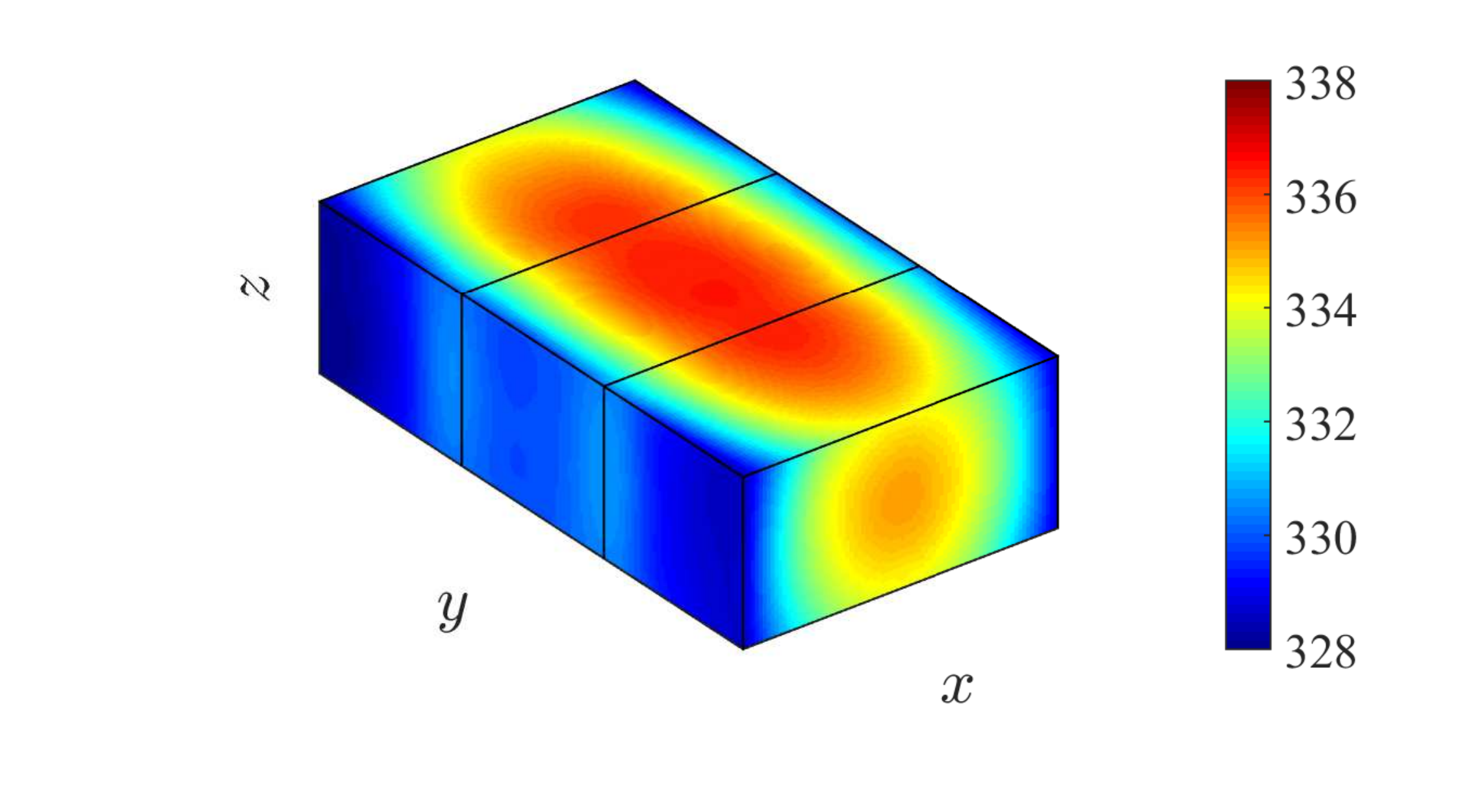}} 
\hspace{0in} 
\caption{{\color{blue}First row: real temperature fields through time for the LiB pack discharged at UDDS-based current loads. Second row: CKF-based reconstruction. Third row: DKF-based reconstruction.} } 
\label{Fig:Overall-Temperature-Fields}
\end{figure*}

\subsection{Simulation Results}\label{Sec:simulation-results}

{\color {blue}Consider the stochastic discrete-time state-space model~\eqref{new-system}. The real temperature field is obtained by running the state equation of~\eqref{new-system} using MATLAB with the effects of process noise included. Sensor-based measurements are obtained according to the measurement equation of~\eqref{new-system}, which are subjected to sensor noise. Using the measurements and based onthe model, the CKF and DKF are applied to reconstruct the temperature field.} The simulation results are summarized in Figure~\ref{Fig:Overall-Temperature-Fields}. The first row shows the true temperature field that evolves over time. Here, trilinear spatial interpolation is used to generate spatially continuous temperature fields. {\color{blue} As is shown, the pack sees an obvious temperature rise, despite the convection cooling and only three cells. In addition, it can be easily found that the temperature differs spatially across the pack, with a high gradient buildup at the end of the simulation. The second row shows the reconstructed temperature fields using the CKF, and the third row shows the ones based on the DKF. It is seen that, although the initial guess differs from the truth, both the CKF and DKF can generate temperature field estimation that gradually catches up with the truth. For both, satisfactory reconstruction starts from around 240 s and is maintained afterwards. Improved estimation accuracy can be expected if the initial temperature guess is made closer to the truth, which is not difficult on many practical occasions, since a LiB pack at rest for a long period will have almost the same temperature with the ambient environment. }

\begin{figure}[t]
\centering
\includegraphics[width=0.5\textwidth]{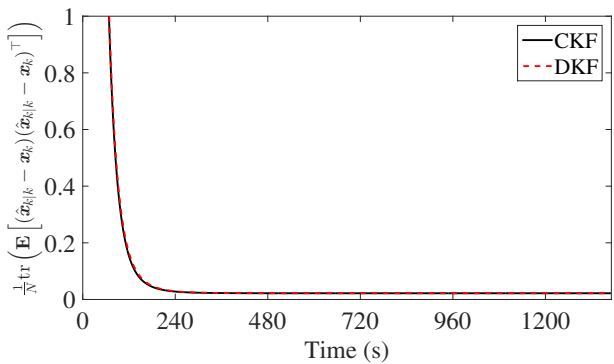}
\centering
\caption{{\color{blue}Evaluation of the estimation error over time.}}
\label{Fig:Covariance-Comparison-CKF-DKF}
\end{figure}

In Figure~\ref{Fig:Overall-Temperature-Fields}, the CKF and DKF present approximately the same estimation results. To further evaluate their accuracy, we consider the following metric:
\begin{align*}
\frac{1}{N} {\textrm{trace}} \left( \mathbf{E}\left[(\hat{\bm{x}}_{k|k}-\bm{x}_k) {(\hat{\bm{x}}_{k|k}-\bm{x}_k)}^{\top}\right ] \right),
\end{align*}
which is the trace of the estimation error covariance averaged over the state space and represents a statistical quantification of the state estimation error. Evaluation of the two algorithms based on this metric is illustrated in Figure~\ref{Fig:Covariance-Comparison-CKF-DKF}. One can see that they both exhibit a decreasing trend over time and remain quite close to each other. This implies that the DKF is quite comparable to the CKF in estimation performance. In addition, as analyzed in Section~\ref{Sec:Complexity-Analysis}, the DKF considerably outperforms the CKF in terms of computational complexity, which thus makes it a more desirable solution for the considered problem of temperature field reconstruction. 

\begin{figure} [t]
\centering 
\includegraphics[trim = {9mm 0mm 0mm 2mm}, clip, width=1.65in]{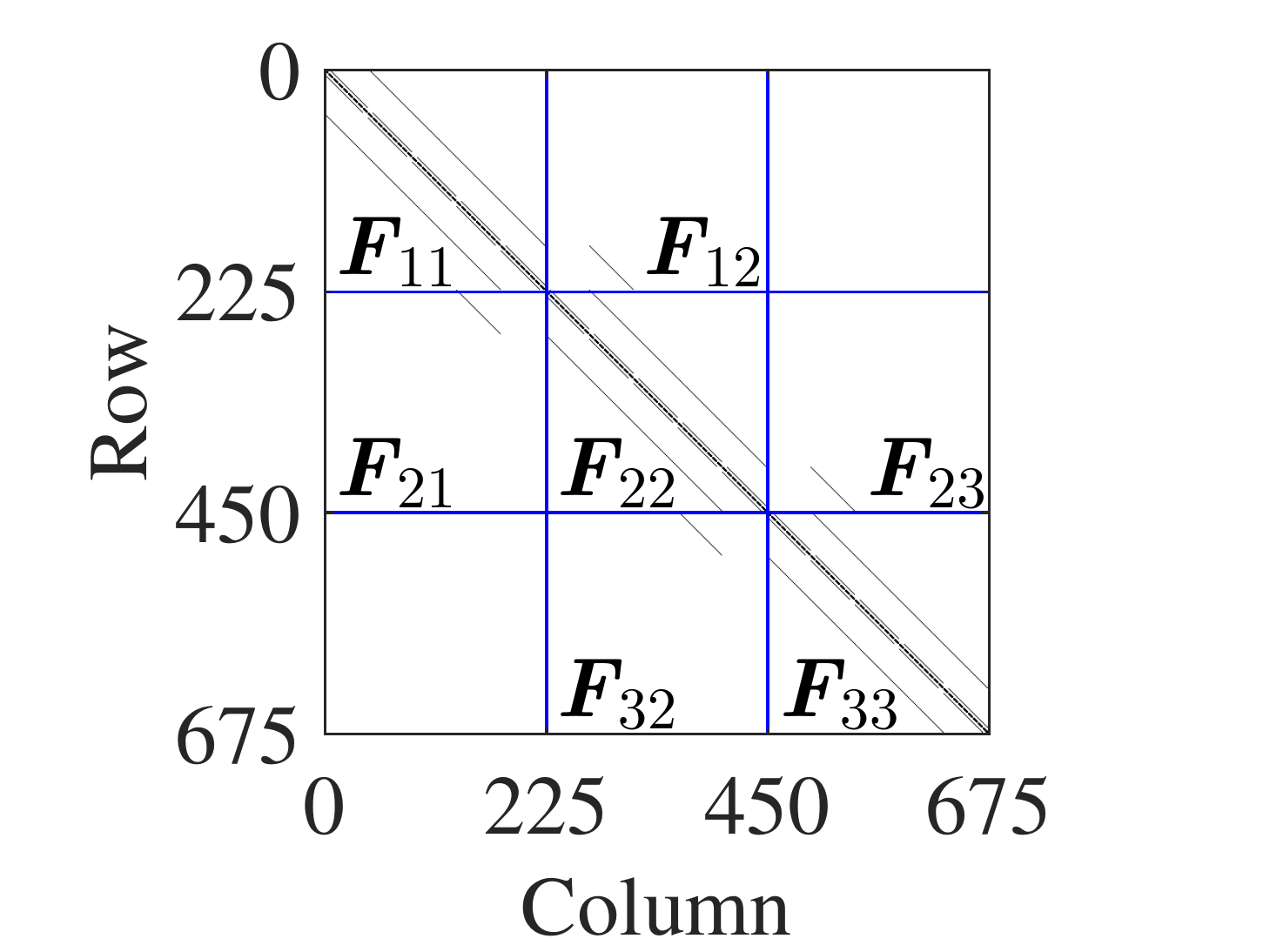}
\caption{Matrix $\bm{F}$ plotted in black and white.} 
\label{matrix-F} 
\end{figure}

{\color{blue}
Next, let us further examine the application of the SS-DKF approach proposed in Section~\ref{Sec:SS-DKF-algorithm}. In this case, the heat generation model~\eqref{heat-generation} is simplified by removing the negligible reversible entropic heating, and consequently, $\bm F$ is time-invariant.} Figure~\ref{matrix-F} offers a visual black-white display of $\bm{F}$, which shows non-zero elements in black and zero elements in white. Figure~\ref{matrix-F} shows that $\bm{F}$ is diagonally dominant, fundamentally ascribed to the pack's serial architecture. It is also found that $\bm{F}_{\mathrm{od}}^2$ equals $\bf 0$, where
\begin{align*}
\bm{F}_{\rm od} = \left[ \begin{array}{ccc}
& \bm{F}_{12} & \\
\bm{F}_{21} &   & \bm{F}_{23} \\
& \bm{F}_{32} & 
\end{array} \right] \in \mathbb{R}^{675\times 675}.
\end{align*}
Specifically, it is verifiable that each row vector of $\bm{F}_{\mathrm{od}}$ multiplied by any column vector $\in \mathbb R^{675\times1}$ in $\bm{F}_{\mathrm{od}}$ gives rise to zero. It can be further verified that this finding is applicable to any LiB pack with the considered configuration. Having established the stability of $\bm{F}$ and zero spectral radius of $\bm{F}_{\rm{od}}$, the SS-DKF algorithm will be asymptotically stable, as discussed in Section~\ref{Sec:SS-DKF-algorithm}. {\color{blue}Now running the CKF, DKF and SS-DKF, one can reconstruct the temperature field. The visual demonstration of the temperature field estimation is omitted to save space. However, a comparison of accuracy is provided in Figure~\ref{Fig:Covariance-Comparison-CKF-DKF-SSDKF}. It is observed that the SS-DKF is less accurate in the initial stage. However, it can achieve approximately the same accuracy after about 400 s. Given this result and its superior computational efficiency, the SS-DKF can be a worthy tool in practice. 
}

\begin{figure}[t]
\centering
\includegraphics[width=0.5\textwidth]{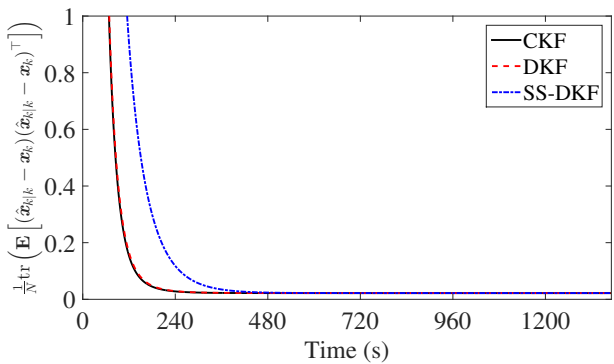}
\centering
\caption{{\color{blue}Evaluation of the estimation error over time when the simplified heat generation model is used.}}
\label{Fig:Covariance-Comparison-CKF-DKF-SSDKF}
\end{figure}

\section{Conclusion}\label{Sec:Conclusion}

This paper considers the real-time reconstruction of the three-dimensional temperature field across a LiB pack. Charging/discharging of LiB cells will generate heat and cause cell-wide temperature rise. With a few cells stacked together, a LiB pack can see stronger heat buildup and spatially distributed temperature. It is known that high temperature can shorten the life of LiBs and may even trigger disastrous fires or explosions. The issue can be more pronounced for packs composed of large-format high-capacity LiB cells, which are in growing use nowadays but have more intense heat generation. An important means of thermal monitoring for safety, the three-dimensional temperature field reconstruction is beneficial and necessary for LiB packs but still absent in the literature.

This paper thus is motivated to address this challenge through model-based temperature estimation. A thermal model is presented first to capture the thermal dynamic process of a LiB pack, which is based on heat transfer and energy balance analysis. Compared to the lumped models and thermo-electrochemical models, such a model achieves a balance between computational efficiency and physical integrity. Based on the model, the well-known KF approach is then distributed to achieve global temperature field reconstruction through localized estimation, reducing the computational complexity remarkably. A DKF algorithm, which is straightforward but well fits with the considered problem, is offered, and its steady-state version, SS-DKF, would require even less computation time. A detailed computational complexity analysis highlights the advantages of the distributed estimation approaches. Simulation with a LiB pack based on genuine cells demonstrates the effectiveness of the DKF and SS-DKF algorithms. The results point to the promises of the proposed methodology for creating more informative thermal monitoring toward improving thermal safety of LiB packs.

{\color{blue}A diversity of work will be performed in the future along this study,  including temperature field reconstruction based on more sophisticated models and heterogeneous cells, DKF-based thermal runaway detection, and experimental validation.}


\balance

\bibliographystyle{IEEEtran} 
\bibliography{reference} 
\end{document}